\documentclass[a4paper,11pt]{article} 
\usepackage[left=3cm,right=3cm,top=3cm,bottom=3cm]{geometry}
\usepackage{braket}
\usepackage[utf8]{inputenc}  
\usepackage{times}  
\usepackage{amsmath, amsthm, amssymb}  
\usepackage{bbm} 
\usepackage{url} 
\usepackage{color}
\usepackage{enumerate}
\usepackage{enumitem}
\usepackage{graphicx} 
\usepackage{mathrsfs} 
\usepackage{stmaryrd}
\usepackage{bm}
\usepackage{amsfonts}
\usepackage{hyperref}
\usepackage{cleveref}
\usepackage{algpseudocode}
\usepackage{mathtools}

\usepackage{ bbold }
\usepackage{ dsfont }

\usepackage{algorithm}
\usepackage{algpseudocode}
\usepackage{mathtools}
\usepackage{relsize,exscale}
\usepackage[hang,small]{caption}
\usepackage{float}
\usepackage{amsthm}
\usepackage{multicol}

\DeclareMathOperator*{\argmax}{arg\,max}
\DeclareMathOperator*{\argmin}{arg\,min}

\theoremstyle{plain} 
\newtheorem{lemma}{Lemma}
\newtheorem{corollaire}[lemma]{Corollary}
\theoremstyle{plain} 
\newtheorem{definition}[lemma]{Definition}

\newtheorem{proposition}[lemma]{Proposition}
\newtheorem*{proposition sans num}{Proposition}

\newtheorem{remark}[lemma]{Remark}
\newtheorem{theorem}[lemma]{Theorem}

\newcommand{\bC}{\mathbbm{C}}
\newcommand{\bF}{\mathbbm{F}}
\newcommand{\bN}{\mathbbm{N}}
\newcommand{\bR}{\mathbbm{R}}

\newcommand{\cC}{\mathcal{C}}
\newcommand{\cD}{\mathcal{D}}
\newcommand{\cE}{\mathcal{E}}
\newcommand{\cF}{\mathcal{F}}
\newcommand{\cG}{\mathcal{G}}

\newcommand{\cO}{\mathcal{O}}
\newcommand{\cP}{\mathcal{P}}
\newcommand{\cQ}{\mathcal{Q}}

\newcommand{\cV}{\mathcal{V}}
\newcommand{\cX}{\mathcal{X}}

\newcommand{\maxDeg}{d_B(d_B + 2 d_A - 1)}
\newcommand\reducedProba{e^{-\Omega(\sqrt{n})}}

\newcommand\fullProbaii{C |\cV| \left(\frac{p}{p_{\textrm{th}}}\right)^{\alpha t}}

\DeclarePairedDelimiter\ceil{\lceil}{\rceil}

\algblockdefx{ForP}{EndForP}[1]{
  \textbf{in parallel for }#1 \textbf{do}
}{
  \textbf{end parallel for}
}

\title{Constant overhead quantum fault-tolerance with quantum expander codes}
\date{}
\author{Omar Fawzi\thanks{Univ Lyon, ENS de Lyon, CNRS, UCBL, LIP UMR 5668, F-69007 Lyon, France. omar.fawzi@ens-lyon.fr} \qquad Antoine Grospellier\thanks{Inria, France. antoine.grospellier@inria.fr, anthony.leverrier@inria.fr} \qquad Anthony Leverrier\footnotemark[2]}

\begin{document} 
\maketitle

\begin{abstract}
We prove that quantum expander codes can be combined with quantum fault-tolerance techniques to achieve constant overhead: the ratio between the total number of physical qubits required for a quantum computation with faulty hardware and the number of logical qubits involved in the ideal computation is asymptotically constant, and can even be taken arbitrarily close to 1 in the limit of small physical error rate.  
This improves on the polylogarithmic overhead promised by the standard threshold theorem. 

To achieve this, we exploit a framework introduced by Gottesman together with a family of constant rate quantum codes, \emph{quantum expander codes}. Our main technical contribution is to analyze an efficient decoding algorithm for these codes and prove that it remains robust in the presence of noisy syndrome measurements, a property which is crucial for fault-tolerant circuits. We also establish two additional features of the decoding algorithm that make it attractive for quantum computation: it can be parallelized to run in logarithmic depth, and is single-shot, meaning that it only requires a single round of noisy syndrome measurement.
\end{abstract}

\section{Introduction}

Quantum computers are expected to offer significant, sometimes exponential, speedups compared to classical computers. For this reason, building a large, universal computer, is a central objective of modern science. Despite two decades of effort, experimental progress has been somewhat slow and the largest computers available at the moment reach a few tens of physical qubits, still quite far from the numbers necessary to run ``interesting'' algorithms. A major source of difficulty is the inherent fragility of quantum information: storing a qubit is quite challenging, processing quantum information even more so. 

Any physical implementation of a quantum computer is necessarily imperfect because qubits are subject to decoherence and physical gates can only be approximately realized. In order to perform a correct computation on a faulty computer, techniques of fault-tolerant computation must be developed. One of the crowning achievements of the early days of quantum computing is the \emph{threshold theorem} which states that upon encoding the logical qubits within the appropriate quantum error correcting code, it is possible to perform arbitrary long computations on a faulty quantum computer, provided that the noise level is below some \emph{constant} threshold value \cite{aharonov2008fault}. This solution comes at a cost, however, since the fault-tolerant version of the circuit to be performed is usually larger that the initial version. In particular, a number of ancilla qubits is required and the space overhead, \textit{i.e.}, the ratio between the total number of physical qubits and the number of logical qubits, scales polylogarithmically with the number of gates involved in the original computation. Indeed, the main technique to achieve fault-tolerance is to protect the logical qubits with concatenated codes. In order to guarantee a final failure probability $\varepsilon$ for a circuit $\mathrm{C}$ acting on $k$ qubits with $|\mathrm{C}|$ locations\footnote{A location refers either to a quantum gate, the preparation of a qubit in a given state, a qubit measurement or a wait location if the qubit is not acted upon at a given time step.}, $\cO(\log \log (|\mathrm{C}|/\varepsilon))$ levels of encoding are needed, which translates into a $\mathrm{polylog}(|\mathrm{C}|/\varepsilon)$ space overhead.
While this might seem like a reasonably small overhead, this remains rather prohibitive in practice, and more importantly, it raises the question of whether this value is optimal. 
In this paper, we consider a realistic model for quantum computing where the quantum gates are noisy, but all classical computation is assumed to be fast and error-free. Note that if classical gates are also noisy, then it is known that classical fault-tolerance cannot be obtained with constant overhead \cite{reischuk1991reliable,gacs1994lower}.

In a breakthrough paper, Gottesman has shown that the polylogarithmic overhead was maybe not necessary after all, and that polynomial-time computations could be performed with a noisy circuit with only a \emph{constant} overhead \cite{gottesman2014fault}. In fact, the constant can even be taken arbitrarily close to 1 provided that the physical error is sufficiently rate small. 
In order to overcome the polylogarithmic barrier, Gottesman suggested to use quantum error correcting codes with \emph{constant rate}.  More precisely, the idea is to encode the logical qubits in large blocks, but still of size sub-linear in $k$. The encoding can still be made fault-tolerant thanks to concatenation, but this only yields an overhead polylogarithmic in the block size, and choosing a sufficiently small block size to a sub-linear overhead for encoding. Gates are then performed with Knill's technique by teleporting the appropriate encoded states. Overall, apart from the initial preparation and final measurement, the encoded circuit alternates between steps applying a gate of the original circuit on the encoded state with Knill's technique, and steps of error correction for the quantum code consisting of a measurement of the syndrome, running a (sufficiently fast) classical decoding algorithm and applying the necessary correction. For this to work out properly, and to keep a constant overhead, the syndrome measurement should be efficient and this will be the case if the quantum code is a low-density parity-check (LDPC) code. Indeed, in that case, the syndrome measurement circuit will be of constant depth and won't require any additional overhead. 
The last property needed for the scheme to work is the existence of an efficient classical decoding algorithm for the quantum code working even in the presence of noise on the syndrome measurement. 

The main result of Gottesman is as follows: provided that the right family of quantum codes exists, it is possible to perform fault-tolerant quantum computing with constant overhead. By right family, we mean a family of constant rate LDPC quantum codes with an efficient decoding algorithm robust against noisy syndrome measurements. More precisely, the decoding algorithm should correct stochastic errors of linear weight and fail with negligible probability. Ideally, we also want the algorithm to be sufficiently fast to avoid errors building up during the decoding. At the time, no such family of codes was known to exist. In fact, families of constant rate LDPC codes with unbounded minimum distance are quite difficult to construct, even forgetting about the decoding problem. 
Possible candidate families include surface codes \cite{freedman2002z2}, 4-dimensional topological codes \cite{guth2014quantum} and hypergraph product codes \cite{tillich2014quantum}. While surface codes come with an efficient decoding algorithm based on minimum weight matching \cite{edmonds1965maximum}, they only display a logarithmic minimum distance if they have constant rate \cite{delfosse2013tradeoffs}. Topological 4-D codes come with a much larger minimum distance, but the available efficient decoding algorithms are only known to perform well for errors of logarithmic weight \cite{hastings2013decoding,londe2017golden}. In both cases, this is insufficient to provide a universal threshold for the error rate, independent of the size of the quantum circuit to be performed. 
Finally, the family of hypergraph product codes yields the best minimum distance to date for constant rate LDPC codes: the minimum distance scales like the square-root of the block length. 
In general, however, we don't know of any efficient decoding algorithm for hypergraph product codes. 

The hypergraph product construction takes a classical code $[n,k,d_{\min}]$ as input an yields a quantum code $[[N, K, D_{\min}]]$ of length $N = \Theta(n^2)$, dimension $K = \Theta(k^2)$ and a minimum distance equal to that of the classical code. When applying this construction with a classical expander code \cite{sipser1996expander}, it yields a so-called \emph{quantum expander code} \cite{leverrier2015quantum}. In that case, one has $K = \Theta(N)$ and $D_{\min} = \Theta(\sqrt{N})$, and interestingly, one can take inspiration of the efficient bit-flip decoding algorithm for classical expander codes \cite{sipser1996expander} to design an efficient decoding algorithm for the quantum expander codes. Such an algorithm, the \emph{small-set-flip decoding algorithm}, was introduced in  \cite{leverrier2015quantum} where it was proved that it corrects arbitrary (adversarial) errors of weight $O(D_{\min})$ in linear time. More recently in~\cite{fawzi2018efficient}, we studied the behavior of this algorithm against stochastic noise and proved that it corrects random errors of linear weight, except with negligible probability.

In the present work, we extend the analysis significantly and show that the algorithm still works in the presence of a noisy syndrome. This was the missing condition to satisfy all the criteria required by Gottesman's construction. In other words, quantum expander codes can be exploited to obtain quantum fault-tolerance with a constant overhead. In addition, we establish two remarkable features of the decoding algorithm: first, it is \emph{single-shot} meaning that the syndrome measurement need not be repeated a polynomial number of times as in typical constructions \cite{bombin2015single}: one measurement suffices; second, the algorithm can be parallelized to run in logarithmic time instead of linear time. This second point is important since storage errors will always affect the qubits during the classical decoding step, meaning that it is crucial to reduce the necessary time as much as possible. We note, however, that for our main result below to hold, we need to assume that the error rate affecting the qubits during the decoding step is constant and doesn't depend on the size of the computation. Without this extra assumption in our model, it is implausible that true constant space overhead quantum fault-tolerance can be achieved. 

We obtain the following general result by using our analysis of quantum expander codes in Gottesman's generic construction~\cite{gottesman2014fault}. More details on the definition of circuits and the noise model are given later in Section~\ref{sec:FT}.
\begin{theorem}
For any $\eta > 1$ and $\varepsilon > 0$, there exists $p_{T}(\eta) > 0$ such that the following holds for sufficiently large $k$. Let $\mathrm{C}$ be a quantum circuit acting on $k$ qubits, and consisting of $f(k)$ locations for $f$ an arbitrary polynomial. There exists a circuit $\mathrm{\tilde{C}}$ using $\eta k$ physical qubits, depth $\cO(f(k))$ and number of locations $\cO(k f(k))$ that outputs a distribution which has total variation distance at most $\varepsilon$ from the output distribution of $\mathrm{C}$, even if the components of $\mathrm{\tilde{C}}$ are noisy with an error rate $p < p_{T}$.
\end{theorem}

Before moving to the proof techniques, let us mention some limitations of the present work. For our analysis to apply, we need bipartite expander graphs with a large (vertex) expansion. A first issue is that there is no known efficient algorithm that can deterministically construct such graphs\footnote{While algorithms to construct graphs with large \textit{spectral} expansion are known, they do not imply a sufficient vertex expansion for our purpose.}. While random graphs will display the right expansion (provided their degree is large enough) with high probability, it is not known how to check efficiently that a given graph is indeed sufficiently expanding. 
The second issue is that we need graphs with a large (constant) degree, which will translate into significantly large quantum codes. In other words, one shouldn't expect the present analysis to be applicable to small size quantum codes that might be built in the near future. We note that Gottesman's analysis also required the initial circuit size to be large enough: this is necessary in order to make the contribution of additive terms sub-linear and therefore obtain a constant overhead. 
Another limitation of our work is the very small threshold value that it yields. While the threshold is usually expected to lie between $10^{-3}$ and $10^{-2}$ for the best constructions based on code concatenation, we expect our value to be several orders of magnitude smaller, as this was already the case in Gottesman's paper \cite{gottesman2014fault} and in our previous work with perfect syndrome measurement \cite{fawzi2018efficient}. Part of the explanation is due to the very crude bounds that we obtain \textit{via} percolation theory arguments. In this work, we haven't tried to optimize the value of the threshold and have instead tried to simplify the general scheme as much as possible. We expect that future work, in particular based on simulations, will help to better understand the true value of the threshold for fault-tolerance schemes with constant overhead. 
Finally, as already pointed out, we consider a model with error-free classical computation, and assume that the logarithmic-depth decoding algorithm can be performed in constant time.

\subsection*{Main result and proof techniques}
In this section, we provide an informal overview of the main result and the techniques used for the proofs. More formal definitions and proofs can be found in the following sections. When decoding a quantum error correcting code, two types of errors need to be taken into account: the $X$-type (or bit flip) errors and the $Z$-type (or phase flip) errors. They play a symmetric role for the codes we consider and it is therefore sufficient to focus on bit flips for instance. An $X$-type error is described by a subset $E$ of the qubits to which the bit flip operator $X$ was applied. To decode a quantum error correcting code, we start by performing a measurement that returns the syndrome $\sigma = \sigma(E)$ which only depends on the error. The objective of the decoding algorithm is given $\sigma$ to output an error $\hat{E}$ which is the same as $E$. More precisely, it is sufficient for the errors $E$ and $\hat{E}$ to be \emph{equivalent} in the sense that the error $E \oplus \hat{E}$ acts trivially on every codeword (for a stabilizer code, this simply means that $E \oplus \hat{E}$ belongs to the stabilizer group). As previously mentioned, our main contribution is to analyze the small-set-flip decoding algorithm in the setting where the syndrome measurement is noisy, \textit{i.e.}, the decoding algorithm takes as input $(\sigma(E) \oplus D)$ instead of just $\sigma(E)$, where $D$ represents the syndrome measurement error. The objective of the decoding algorithm is then not to recover the error exactly (which will not be possible) but rather to control the size of remaining error $E \oplus \hat{E}$. In the context of quantum fault-tolerance, the relevant error model for the pair $(E, D)$ is the \emph{local stochastic noise model} with parameters $(p, q)$ defined by requiring that for any $F$ and $G$, the probability that $F$ and $G$ are part of the qubit and syndrome errors, respectively, is bounded as follows, $\mathbb{P}[F \subseteq E, G \subseteq D] \leq p^{|F|} q^{|G|}$.

\begin{theorem}[Informal]
  There exist constants $p_0 > 0, p_1 > 0$ such that the following holds. Consider a bipartite graph with sufficiently good expansion and the corresponding quantum expander code. Consider random errors $(E, D)$ satisfying a local stochastic noise model with parameter $(p_{\mathrm{phys}}, p_{\mathrm{synd}})$ with $p_{\mathrm{phys}} < p_0$ and $p_{\mathrm{synd}} < p_1$. Let $\hat{E}$ be the output of the small-set-flip decoding algorithm on the observed syndrome. Then, except for a failure probability of $\reducedProba$, $E \oplus \hat{E}$ is equivalent to $E_{\textrm{ls}}$ that has a local stochastic distribution with parameter $p_{\mathrm{synd}}^{\Omega(1)}$. In addition, the small-set-flip algorithm can be parallelized to run in $\cO(\log n)$ depth.
\end{theorem}
In the special case where the syndrome measurements are perfect, \textit{i.e.}, $p_{\mathrm{synd}} = 0$, the statement guarantees that for a typical error of size at most $p_0 n$, the small-set-flip algorithm finds an error that is equivalent to the error that occurred. If the syndrome measurements are noisy, then we cannot hope to recover an equivalent error exactly, but instead we can control the size of the remaining error $E \oplus \hat{E}$ by the amount of noise in the syndrome measurements. In particular, for any qubit error rate below $p_0$, the decoding operation reduces this error rate to be $p_{\mathrm{synd}}^{\Omega(1)}$ (our choice of $p_0$ will be such that $p_{\mathrm{synd}}^{\Omega(1)} \ll p_0$). This criterion is sufficient for fault-tolerant schemes as it ensures that the size of the qubit errors stay bounded throughout the execution of the circuit.
The proof of this theorem consists of two main parts: analyzing arbitrary low weight errors below the minimum distance (\Cref{prop:error_qubits_syn}) and exploiting percolation theory to analyze stochastic errors of linear weight (\Cref{thm:correction}). 

The small-set-flip decoding algorithm proceeds by trying to flip small sets of qubits so as to decrease the weight of the syndrome, and the main challenge in its analysis is to prove the existence of such a small set $F$. In the case where the observed syndrome is error free, Refs \cite{leverrier2015quantum} and \cite{fawzi2018efficient} relied on the existence of a ``critical generator'' to exhibit such a set of qubits. This approach, however, only yields a \emph{single} such set $F$, and when the syndrome becomes noisy, nothing guarantees anymore that flipping the qubits in $F$ will result in a decrease of the syndrome weight and it becomes unclear whether the decoding algorithm can continue. Instead, in order to take into account the errors on the syndrome measurements, we would like to show that there are \emph{many} possible sets of qubits $F$ that decrease the syndrome weight. In order to establish this point, we consider an error $E$ of size below the minimum distance and we imagine running the (sequential) decoding algorithm \cite{leverrier2015quantum} without errors on the syndrome. The algorithm gives a sequence of small sets $\{F_i\}$ to flip successively in order to correct the error. In other words, we obtain the following decomposition of the error, $E = \oplus_i F_i$ (note that the sets $F_i$ might overlap). The expansion properties of the graph guarantee that there are very few intersections between the syndromes $\sigma(F_i)$ (see the proof of \Cref{prop:error_qubits_syn}). In particular in the case of noiseless syndrome, a linear number of these $F_i$ can be flipped to decrease the syndrome weight. There are two consequences to this result. First, it is possible to parallelize the decoding algorithm by flipping multiple $F_i$ in each round (\Cref{sec:parallel}) and this decreases the syndrome weight by a constant factor, thereby correcting the error after a logarithmic number of rounds. Second, even when the syndrome is noisy there will remain some $F_i$ that can be flipped in order to decrease the syndrome weight and finally, the size of the error $E \oplus \hat{E}$ can be upper bounded with a linear function of the syndrome error size (\Cref{prop:error_qubits_syn}).

In order to analyze random errors of linear weight, we show using percolation theory that, with high probability, the error forms clusters in the sense of connected $\alpha$-subsets (\Cref{def:alpha} and \Cref{lem:perco}). This is similar to the analysis in~\cite{kovalev2013fault,fawzi2018efficient}, except that we use the syndrome adjacency graph of the code (as in~\cite{gottesman2014fault}) to establish the ``locality'' of the decoding algorithm, implying that each cluster of the error is corrected independently of the other ones (\Cref{lem:local}). Using the fact that clusters are of size bounded by the minimum distance of the code, the result on low weight errors shows that the size of $E \oplus \hat{E}$ is controlled by the syndrome error size. In order to show that the error after correction is local stochastic, we introduce the notion of \emph{witness} (\Cref{def:witness}). The basic idea is to find a syndrome error in the neighborhood of a given qubit error $S$. However, a qubit error can be the consequence of a distant syndrome error. This is why a witness is defined as a set of qubit errors $W$ but is potentially larger than $S$. The previously mentioned results show that witnesses exist for $E \oplus \hat{E}$ and we conclude our proof using an upper bound on the probability that a witness exists. 

A remarkable feature of our analysis is that it shows that the small-set-flip decoding algorithm only uses a single noisy syndrome measurement and outputs an error with controlled weight. Note that this is in contrast to decoding algorithms for many other codes such as the toric code for which such a repetition is necessary. This property is called \emph{single-shot} in the fault-tolerant quantum computation literature \cite{bombin2015single,campbell2018theory}.

%%%%%%%%%%%%%%%%%%%%%%
\subsection*{Organization}

We start in Section \ref{sec:prelim} with notations and preliminaries to recall the construction and main properties of quantum expander codes, their efficient decoding algorithm. We also introduce the relevant noise model and give an overview of the fault-tolerant scheme of~\cite{gottesman2014fault} applied to quantum expander codes. In Section~\ref{sec:analysis}, we establish our main technical result showing that the small-set-flip decoding algorithm for quantum expander codes is robust against local stochastic noise in the syndrome measurement. In Section~\ref{sec:parallel}, we prove that the algorithm can be efficiently parallelized and works in logarithmic time.

%%%%%%%%%%%%%%%%%%%%%%
\section{Preliminaries}
\label{sec:prelim}

In this section, we first review the construction of classical and quantum expander codes. 
We then discuss models of noise which are relevant in the context of quantum fault-tolerance. Finally, we describe the fault-tolerant construction of Gottesman applied to quantum expander codes.

\subsection{Classical expander codes}
\label{subsec:classical}

A linear classical error correcting code $\cC$ of dimension $k$ and length $n$ is a subspace of $\bF_2^n$ of dimension $k$. Mathematically, it can be defined as the kernel of an $(n-k) \times n$ matrix $H$, called the parity-check matrix of the code: $\cC = \{ x \in \bF_2^n \: : \: Hx = 0\}$. The minimum distance $d_{\min}$ of the code is the minimum Hamming weight of a nonzero codeword: $d_{\min} = \min \{ |x| \: : \: x \in \cC, x\ne 0\}$. Such a linear code is often denoted as $[n,k,d_{\min}]$. 
It is natural to consider families of codes, instead of single instances, and study the dependence between the parameters $n, k$ and $d_{\min}$. In particular, a family of codes has \emph{constant rate} if $k = \Theta(n)$. 
Another property of interest of a linear code is the weight of the rows and columns of the parity-check matrix $H$. If these weights are upper bounded by a constant, then we say that the code is a \emph{low-density parity-check} (LDPC) code~\cite{gallager1962low}. This property is particularly attractive because it allows for efficient decoding algorithms, based on message passing for instance. 

An alternative description of a linear code is \textit{via} a bipartite graph known as its \emph{factor graph}. Let $G =(V \cup C, \cE)$ be a bipartite graph, with $|V|=n_V$ and $|C|=n_C$. With such a graph, we associate the $n_C \times n_V$ matrix $H$, whose rows are indexed by the vertices of $C$, whose columns are indexed by the vertices of $V$, and such that $H_{cv} = 1$ if $v$ and $c$ are adjacent in $G$ and $H_{cv}=0$ otherwise. The binary linear code $\cC_G$ associated with $G$ is the code with parity-check matrix $H$. The graph $G$ is the \emph{factor graph} of the code $\cC_G$ , $V$ is the set of \emph{bits} and $C$ is the set of \emph{check-nodes}.

It will sometimes be convenient to describe codewords and error patterns as subsets of $V$: the binary word $e \in \bF_2^{n_V}$ is described by a subset $E \subseteq V$ whose indicator vector is $e$. Similarly we define the \emph{syndrome} of a binary word either as a binary vector of length $n_C$ or as a subset of $C$:
  \begin{align*}
    \sigma(e) := H e \in \bF_2^{n_C},
    && \sigma(E) := \bigoplus_{v \in E} \Gamma(v) \subseteq C,
  \end{align*}
  where $\Gamma(v) \subseteq C$ is the set of neighbors of $v$.
  In this paper, the operator $\oplus$ is interpreted either as the symmetric difference of sets or as the bit-wise exclusive disjunction depending on whether errors and syndromes are interpreted as sets or as binary vectors.

A family of codes that will be central in this work are those associated to so-called \emph{expander graphs}, that were first considered by Sipser and Spielman in \cite{sipser1996expander}.

\begin{definition}[Expander graph]
  \label{def:exp}
  Let $G = (V \cup C, \cE)$ be a bipartite graph with left and right degrees bounded by $d_V$ and $d_C$ respectively. Let $|V| = n_V$ and $|C| = n_C$. We say that $G$ is $(\gamma, \delta)$-\emph{left-expanding} for some constants $\gamma, \delta >0$, if for any subset $S \subseteq V$ with $|S| \leq \gamma n_V$, the neighborhood $\Gamma(S)$ of $S$ in the graph $G$ satisfies $|\Gamma(S)| \geq (1-\delta) d_V |S|$. Similarly, we say that $G$ is $(\gamma, \delta)$-\emph{right-expanding} if for any subset $S \subseteq C$ with $|S| \leq \gamma n_C$, we have $|\Gamma(S)| \geq (1-\delta) d_C |S|$. Finally, the graph $G$ is said $(\gamma, \delta)$-\emph{expanding} if it is both $(\gamma, \delta)$-left expanding and $(\gamma, \delta)$-right expanding.  
\end{definition}

Sipser and Spielman introduced \emph{expander codes}, which are the linear codes associated with (left-)expander graphs. Remarkably these codes come with an efficient decoding algorithm that can correct \emph{arbitrary} errors of weight $\Omega(n)$ \cite{sipser1996expander}.
\begin{theorem}[Sipser, Spielman \cite{sipser1996expander}]
  \label{thm:SS}
  Let $G = (V \cup C, \cE)$ be a $(\gamma, \delta)$-left expander graph with $\delta < 1/4$. There exists an efficient decoding algorithm for the associated code $\cC_G$ that corrects all error patterns $E \subseteq V$ such that $|E| \leq \gamma (1-2\delta) |V|$.
\end{theorem}
Note that the relevant notion of expansion here is that of \emph{vertex expansion}, which differs from \emph{spectral expansion} (defined as a property of the adjacency matrix of the graph). While the two notions are related, we emphasize that even optimal spectral expansion (as in Ramanujan graphs) doesn't imply values of $\delta$ below $1/2$ and are therefore insufficient for our purpose.

The decoding algorithm called the ``bit-flip'' algorithm is very simple: one simply cycles through the bits in $V$ and flip them if this operation leads to a reduction of the syndrome weight. Sipser and Spielman showed that provided the expansion is sufficient, such an algorithm will always succeed in identifying the error if its weight is below $\gamma(1-2\delta)|V|$. 
In this paper, however,  we will be interested in the decoding of \emph{quantum} expander codes, that we will review next.

Before that, let us mention for completeness that although finding explicit constructions of highly-expanding graphs is a hard problem, such graphs can nevertheless be found efficiently by probabilistic techniques. Verifying that a given graph is expanding is a hard task, however.
\begin{theorem}[Theorem 8.7 of \cite{richardson2008modern}]
  \label{thm:exist}
  Let $\delta$ be a positive constant. For integers $d_A > 1/\delta$ and $d_B > 1/\delta$, a graph $G = (A \cup B, \cE)$ with left-degree bounded by $d_A$ and right-degree bounded by $d_B$ chosen at random according to some distribution is $(\gamma, \delta)$-expanding for $\gamma = \Omega(1)$ with high probability.
\end{theorem}

\subsection{Quantum error correcting codes}

A quantum code encoding $k$ logical qubits into $n$ physical qubits is a subspace of $(\bC^2)^{\otimes n}$ of dimension $2^k$. A quantum \emph{stabilizer code} is described by a stabilizer, that is an Abelian group of $n$-qubit Pauli operators (tensor products of single-qubit Pauli operators $X = \left( \begin{smallmatrix} 0 & 1\\1&0 \end{smallmatrix}\right), Y=ZX, Z=\left( \begin{smallmatrix} 1 & 0\\0&-1 \end{smallmatrix}\right)$ and $I$ with an overall phase of $\pm1$ or $\pm i$) that does not contain $- I$. The code is defined as the eigenspace of the stabilizer with eigenvalue $+1$ \cite{gottesman1997stabilizer}.
A stabilizer code of dimension $k$ can be described by a set of $n-k$ generators of its stabilizer group.

A particularly nice construction of stabilizer codes is given by the CSS construction \cite{calderbank1996good}, \cite{steane1996error}, where the stabilizer generators are either products of single-qubit $X$-Pauli matrices or products of $Z$-Pauli matrices. The condition that the stabilizer group is Abelian therefore only needs to be enforced between $X$-type generators (corresponding to products of Pauli $X$-operators) and $Z$-type generators. 
More precisely, consider two classical linear codes $\cC_X$ and $\cC_Z$  of length $n$ satisfying $\cC_Z^\perp \subseteq \cC_X$ (or equivalently, $\cC_X^\perp \subseteq \cC_Z$), where the dual code $\cC_X^\perp$ to $\cC_X$ consists of the words which are orthogonal to all the words in $\cC_X$. This condition also reads $H_X \cdot H_Z^T=0$, if $H_X$ and $H_Z$ denote the respective parity-check matrices of $\cC_X$ and $\cC_Z$.
The quantum code $CSS(\cC_X, \cC_Z)$ associated with $\cC_X$ (used to correct $X$-type errors and corresponding to $Z$-type stabilizer generators) and $\cC_Z$ (used to correct $Z$-type errors and corresponding to $X$-type stabilizer generators) has length $n$ and is defined as the linear span of $\left\{ \sum_{z \in \cC_Z^\perp} |x +z \rangle \: : \: x \in \cC_X\right\}$, where $\{ |x\rangle \: : \: x\in \bF_2^n\}$ is the canonical basis of $(\bC^2)^{\otimes n}$.
In particular, two states differing by an element of the stabilizer group are equivalent. The dimension of the CSS code is given by $k = \dim (\cC_X/ \cC_Z^\perp) = \dim (\cC_Z/\cC_X^\perp) = \dim \cC_X+ \dim \cC_Z - n$. Its minimum distance is defined in analogy with the classical case as the minimum number of single-qubit Pauli operators needed to map a codeword to an orthogonal one. For the code $CSS(\cC_X, \cC_Z)$, one has $d_{\min} = \min(d_X,d_Z)$ where $d_X = \min \{|E|: E \in \cC_X \setminus {\cC_Z}^{\bot} \}$ and $d_Z = \min \{|E|: E \in \cC_Z \setminus {\cC_X}^{\bot} \}$. We say that  $CSS(\cC_X, \cC_Z)$ is a $[[n,k,d_{\min}]]$ quantum code. In the following, it will be convenient to consider the factor graph $G_X = (V \cup C_X, \cE_X)$ (resp.~$G_Z$) of $\cC_X$ (resp.~of $\cC_Z$). We will denote by  $\Gamma_X$ (resp.~$\Gamma_Z$) the neighborhood in $G_X$ (resp.~$G_Z$). For instance, if $g \in C_Z$ is an $X$-type generator, that is a product of Pauli $X$ operators, then $\Gamma_Z(g)$ is the set of qubits (indexed by $V$) on which the generator acts non-trivially.

Among stabilizer codes, and CSS codes, the class of quantum LDPC codes stands out for practical reasons: these are the codes for which one can find \emph{sparse} parity-check matrices $H_X$ and $H_Z$. More precisely, such matrices are assumed to have constant row weight and constant column weight. Physically, this means that each generator of the stabilizer acts at most on a constant number of qubits, and that each qubit is acted upon by a constant number of generators. Note, however, that while surface codes exhibit in addition spatial locality in the sense that interactions only involve spatially close qubits (for an explicit layout of the qubits in Euclidean space), we do not require this for general LDPC codes. This means that generators might involve long-range interactions. This seems necessary in order to find constant rate quantum codes with growing minimum distance~\cite{bravyi2009no}.

A natural noise model is the so-called \emph{Pauli-type noise}, mapping a qubit $\rho$ to $p_{\mathbbm{1}} \rho + p_X X \rho X + p_Y Y \rho Y + p_Z Z \rho Z$, for some $p_{\mathbbm{1}}, p_X, p_Y, p_Z$. Such a noise model is particularly convenient since one can interpret the action of the noise as applying a given Pauli error with some probability. As usual, it is sufficient to deal with both $X$ and $Z$-type errors in order to correct Pauli-type errors, and one can therefore define an error by the locations of the Pauli $X$ and Pauli $Z$ errors.
An \emph{error pattern} is a pair $(E_X, E_Z)$ of $n$-bit strings, which describe the locations of the Pauli $X$ errors, and Pauli $Z$ errors respectively. The syndrome associated with $(E_X, E_Z)$ for the code $CSS(\cC_X, \cC_Z)$ consists of $\sigma_X = \sigma_X(E_X) := H_X E_X$ and $\sigma_Z = \sigma_Z(E_Z) := H_Z E_Z$. 
A decoder is given the pair $(\sigma_X, \sigma_Z)$ of syndromes and should return a pair of errors $(\hat{E}_X, \hat{E}_Z)$ such that $E_X + \hat{E}_X \in \cC_Z^\perp$ and $E_Z + \hat{E}_Z \in \cC_X^\perp$. In that case, the decoder outputs an error equivalent to $(E_X, E_Z)$, and we say that it succeeds.
\\Similarly as in the classical case, it will be convenient to describe $X$-type error patterns and $X$-type syndromes as subsets of the vertices of the factor graph $G_X=(V \cup C_X, \cE_X)$. The error pattern is then described by a subset $E_X \subseteq V$ whose syndrome is the subset $\sigma_X(E_X) \subseteq C_X$ defined by $\sigma_X(E_X) := \bigoplus_{v \in E_X} \Gamma_X(v)$. One describes $Z$-type error patterns and $Z$-type syndromes in the same fashion using the factor graph $G_Z$.

In this paper, we consider Algorithm \ref{algo decodage qec} which tries to recover $E_{X}$ and $E_{Z}$ independently. More precisely, the algorithm is given by an $X$-decoding algorithm that takes as input $\sigma_{X}$ and returns $\hat{E}_X$ such that $\sigma_X(\hat{E}_X) = \sigma_X$, and a $Z$-decoding algorithm that takes as input $\sigma_{Z}$ and returns $\hat{E}_Z$ such that $\sigma_Z(\hat{E}_Z) = \sigma_Z$. Here the two algorithms are identical upon exchanging the roles of $X$ and $Z$. We note that this kind of decoding algorithm might achieve sub-optimal error probabilities for some error models. In fact, if there are correlations between $X$ and $Z$ errors (for instance in the case of the depolarizing channel where $p_X = p_Y =p_Z$), one can decrease the error probability by trying to recover $E_{X}$ by using both $\sigma_{X}$ and $\sigma_{Z}$.

Let us conclude this section by mentioning the generic (but inefficient) algorithm which returns an error $(\hat{E}_X, \hat{E}_Z)$ of minimum Hamming weight with the appropriate syndrome \footnote{For degenerate quantum codes, this might differ from the output of \emph{maximum likelihood decoding algorithm} which performs an optimization over the probability of equivalent errors.}, that is:
\begin{align*}
  \hat{E}_X = \argmin_{\sigma_X(F_X) = \sigma_X} |F_X|, \quad \hat{E}_Z = \argmin_{ \sigma_Z(F_Z) = \sigma_Z} |F_Z|.
\end{align*}
This algorithm always succeeds provided that the error weights satisfy $|E_X| \leq \left\lfloor (d_X-1)/2\right\rfloor$ and $|E_Z| \leq \left\lfloor (d_Z-1)/2\right\rfloor$.

\subsection{Quantum expander codes}\label{subs:qec}

In this work, we are particularly interested in a family of LDPC CSS codes that features a constant rate and a minimum distance $\Theta(\sqrt{n})$ obtained by applying the hypergraph product construction of Tillich and Z\'emor to classical expander codes. If these expander codes have sufficient expansion, the corresponding quantum code is called \emph{quantum expander code} and comes with an efficient decoding algorithm  \cite{leverrier2015quantum}. 

The construction is as follows. Let $G = (A \cup B, \cE)$ be a biregular  $(\gamma, \delta)$-expanding graph with $\delta$ sufficiently small \footnote{The existence of an efficient algorithm that corrects arbitrary errors of size $O(\sqrt{n})$ is guaranteed as soon as $\delta < 1/6$ \cite{leverrier2015quantum}. The same algorithm corrects random errors of linear weight except with negligible probability as soon as $\delta < 1/8$ \cite{fawzi2018efficient} and in the present paper, we will require the more stringent condition  $\delta < 1/16$.}, and constant left and right degrees denoted $d_A$ and $d_B$ with $d_A \leq d_B$. Let us also denote $n_A = |A|$ and $n_B = |B|$ with $n_B \leq n_A$. According to Theorem \ref{thm:exist}, such graphs can be found in a probabilistic fashion provided that $d_A \geq \lceil \delta^{-1} \rceil$.
Let $\cC$ be the classical code associated with $G$, let $d_{\min}(\cC)$ be the minimal distance of $\cC$ and let $H$ be its parity-check matrix (that we assume to be full rank) corresponding to the factor graph $G$. In particular, the weights of rows and columns of $H$ are $d_A$ and $d_B$, respectively. The hypergraph product code of $\cC$ with itself admits the following parity check matrices:
\begin{align*}
  H_X &= \left( I_{n_A} \otimes H, H^T \otimes I_{n_B}\right),\\
  H_Z &= \left( H \otimes I_{n_A}, I_{n_B} \otimes H^T \right).
\end{align*}
It is immediate to check that this defines a legitimate CSS code since 
\begin{align*}
  H_X H_Z^T &=  I_{n_A} \otimes H \cdot ( H \otimes I_{n_A})^T + H^T \otimes I_{n_B} \cdot (I_{n_B} \otimes H^T)^T \\
  &= H^T \otimes  H+ H^T \otimes H =0. 
\end{align*}
Moreover, the code is LDPC with generators of weight $d_A + d_B$ and qubits involved in at most $2 d_B$ generators.
\\We can describe the factor graphs $G_X$ and $G_Z$ as follows: the set of qubits is indexed by $V := A^2 \cup B^2$, the set of $Z$-type generators is indexed by $C_X := A \times B$ and the set of $X$-type generators is indexed by $C_Z := B \times A$. The bipartite graph $G_X$ has left vertices $V$, right vertices $C_X$ and there is an edge between a vertex $(\alpha,a) \in A^2$ (resp. $(b,\beta) \in B^2$) and a vertex $(\alpha,\beta) \in A \times B$ when $a$ (resp. $b$) is in the neighborhood of $\beta$ (resp. $\alpha$) in $G$. The bipartite graph $G_Z$ has left vertices $V$, right vertices $C_Z$ and there is an edge between a vertex $(\alpha,a) \in A^2$ (resp. $(b,\beta) \in B^2$) and a vertex $(b,a) \in B \times A$ when $\alpha$ (resp. $\beta$) is in the neighborhood of $b$ (resp. $a$) in $G$.
\\The following theorem summarizes the main properties of this quantum code.
\begin{theorem}[Tillich, Z\'emor \cite{tillich2014quantum}]\label{thm:TZ}
  The CSS code defined above is LDPC with parameters $\left[\left[ n, k, d_{\min}\right]\right]$, where $n = {n_A}^2 + {n_B}^2, k \geq (n_A - n_B)^2$ and $d_{\min} = d_{\min}(\cC)$.
\end{theorem}

A natural approach to perform error correction would be to directly mimic the classical bit-flip decoding algorithm analyzed by Sipser and Spielman, that is try to apply $X$-type (or $Z$-type) correction to qubits when it leads to a decrease of the syndrome weight. Unfortunately, in that case, there are error configurations of constant weight that couldn't be corrected: these correspond to error patterns consisting of half a generator where exactly half of its qubits in $A\times A$ and half of its qubits in $B\times B$ are in error. It is easy to check that flipping any single qubit of that generator leaves the syndrome weight invariant, equal to $d_A d_B/2$. On the other hand, removing all the errors from that generator decreases the syndrome weight by $d_A d_B/2$. This suggests the ``small-set-flip'' strategy that we describe next.

Focusing on $X$-type errors for instance, and assuming that the syndrome $\sigma = H_X E$ is known, the algorithm cycles through all the $X$-type generators of the stabilizer group (\textit{i.e.} the rows of $H_Z$), and for each one of them, determines whether there is an error pattern contained in the generator that decreases the syndrome weight.
Assuming that this is the case, the algorithm applies the error pattern (choosing the one maximizing the ratio between the syndrome weight decrease and the pattern weight), if there are several). The algorithm then proceeds by examining the next generator. Since the generators have (constant) weight $d_A+d_B$, there are $2^{d_A + d_B} = \mathcal{O}(1)$ possible patterns to examine for each generator. 
 
Before describing the algorithm more precisely, let us introduce some additional notations. 
Let $\cX$ be the set of subsets of $V$ corresponding to $X$-type generators: $\cX = \{\Gamma_Z(g): g \in C_Z\} \subseteq \cP(V)$, where $\cP(V)$ is the power set of $V$. The indicator vectors of the elements of $\cX$ span the dual code $\cC_Z^\perp$. 
  The condition for successful decoding of the $X$ error $E$ then asks that there exists a subset $X \subset \cX$ such that
  \begin{align*}
    E \oplus \hat{E} = \bigoplus_{x \in X} x, 
  \end{align*}
  where $\hat{E}$ is the output of the decoding algorithm. This means that the remaining error after decoding is trivial, that is equal to a sum of generators.
    At each step, the small-set-flip algorithm tries to flip a subset of $\Gamma_Z(g)$ for some generator $g \in C_Z$ which decreases the syndrome weight $|\sigma|$. In other words, it tries to flip some element $F \in \cF_0$ such that $\Delta(\sigma, F) > 0$ where:
    \begin{align}\label{eq:notations1}
      \cF_0 := \{F \subseteq \Gamma_Z(g): g \in C_Z\},
      && \Delta(\sigma, F) := |\sigma| - |\sigma \oplus \sigma_X(F)|.
    \end{align}

The decoding algorithm consists of two iterations of Algorithm \ref{algo decodage qec0} below: it first tries to correct $X$-type errors, then it is applied a second time (exchanging the roles of $X$ and $Z$) to correct $Z$-type errors.
  \begin{algorithm}[H]
    \caption{(Ref.~\cite{leverrier2015quantum}): Small-set-flip decoding algorithm for quantum expander codes
    }\label{algo decodage qec0}
      {\bf INPUT:} $\sigma \subseteq {C}_X$, a syndrome where $\sigma = \sigma_X(E)$ with $E \subseteq {V}$ an error
      \\{\bf OUTPUT:} $\hat{E}\subseteq {V}$, a guess for the error pattern (alternatively, a set of qubits to correct)
      \\{\bf SUCCESS:} if $E \oplus \hat{E} = \bigoplus_{x \in X} x$ for $X \subseteq \cX$, \textit{i.e.} $E$ and $\hat{E}$ are equivalent errors
      
    \hrule
      \begin{algorithmic}
        \State{$\hat{E}_0 = 0$ ; $\sigma_0 = \sigma$ ; $i = 0$}
        \While{$\displaystyle \left(\exists F \in \cF_0:  \Delta(\sigma_i, F) > 0\right)$}
        \\\State{$\displaystyle F_i = \argmax_{F \in \cF_0} \frac{\Delta(\sigma_i, F)}{|F|} \qquad$ // pick an arbitrary one if there are several choices}
        \State{$\hat{E}_{i+1} = \hat{E}_i \oplus F_i$}
        \State{$\sigma_{i+1} = \sigma_i \oplus \sigma_X (F_i)$ \qquad // $\sigma_{i+1} = \sigma_X (E \oplus \hat{E}_{i+1})$}
        \State{$i = i+1$}
        \EndWhile
        \State{\Return{$\hat{E}_i$}}
      \end{algorithmic}
  \end{algorithm}
 
It was proven in Ref.~\cite{leverrier2015quantum} that this algorithm corrects arbitrary errors of size $O(\sqrt{n})$ provided that the expansion of the graph satisfies $\delta < 1/6$.

\begin{theorem}[Leverrier, Tillich, Z\'emor \cite{leverrier2015quantum}]\label{thm:LTZ}
  Let $G = (A \cup B, \cE)$ be a $(d_A,d_B)$-biregular  $(\gamma, \delta)$-expanding graph with $\delta < 1/6$. Letting $d_A$ and $d_B$ be fixed and allowing $n_A, n_B$ to grow, the small-set-flip decoding algorithm (\Cref{algo decodage qec0}) runs in time linear in the code length $n=n_A^2 + n_B^2$, and decodes any quantum error pattern of weight less than
  \begin{align}
    \label{eq:w0}
    w_0 = \frac{\gamma n_B}{3(1+d_B)}.
  \end{align}
\end{theorem}

In a recent work, the analysis was extended to the case of random errors (either independent and identically distributed, or local stochastic) provided that the syndrome extraction is performed perfectly and under a stricter condition on the expansion of the graph \cite{fawzi2018efficient}.

  \begin{theorem}[Fawzi, Grospellier, Leverrier \cite{fawzi2018efficient}] \label{main theorem}
     Let $G = (A \cup B, \cE)$ be a $(d_A,d_B)$-biregular  $(\gamma, \delta)$-expanding graph with $\delta < 1/8$. Then there exists a probability $p_0 > 0$ and constants $C,C'$ such that if the noise parameter on the qubits satisfies $p < p_0$, the small-set-flip decoding algorithm (\Cref{algo decodage qec0}) runs in time linear in the code length and corrects a random error with probability at least $1 - C n\left(\frac{p}{p_0}\right)^{C' \sqrt{n}}$.
  \end{theorem}

The caveat of this result is that it only applies in absence of errors for the syndrome extraction. The main technical contributions of this paper are to establish that the same algorithm still works in presence of noise on the syndrome, and to show that the decoding algorithm can be parallelized to run in logarithmic time.

\subsection{Noise models}
\label{sec:noise}

In the context of quantum fault-tolerance, we are interested in modeling noise occurring during a quantum computation. We refer the reader to the introduction on the topic by Gottesman for a thorough description of noise models for fault-tolerance \cite{gottesman2009introduction}.
In the circuit model of quantum computation, the effect of noise is to cause faults occurring at different locations of the circuit: either on the initial state and ancillas, on gates (either active gates or storage gates) or on measurement gates. We refer to this model as \emph{basic model} for fault-tolerance. The main idea to perform a computation in a fault-tolerant manner is then to encode the logical qubits with a quantum error correcting code, replace the locations of the original circuit by gadgets applying the corresponding gate on the encoded qubits, and interleave the steps of the computation with error correction steps. 
In general, it is convenient to abstract away the details of the implementation and consider a \emph{simplified model} of fault-tolerance where one is concerned with only two types of errors: errors occurring at each time step on the physical qubits, and errors on the results of the syndrome measurement. The link between the basic and the simplified models for fault-tolerance can be made once a specific choice of gate set and gadgets for each gate is made. This is done for instance in Section 7 of Ref.~\cite{gottesman2014fault}.
In other words, the simplified model of fault-tolerance allows us to work with quantum error correcting codes where both the physical qubits and the check nodes are affected by errors. 

As usual in the context of quantum error correction, we restrict our attention to Pauli-type errors since the ability to correct all Pauli errors of weight $t$ implies that arbitrary errors of weight $t$ can be corrected. In particular, one only needs to address $X$ and $Z$-type errors since a $Y$-error corresponds to simultaneous $X$ and $Z$-errors. Therefore, we think of an error pattern on the qubits as a pair $(E_X, E_Z)$ of subsets of the set of qubits $V$. This should be interpreted as Pauli error $X$ on all qubits in $E_X \setminus E_Z$, error $Y$ on $E_X \cap E_Z$ and error $Z$ on $E_Z \setminus E_X$.
Similarly, the error on the syndrome consists of two classical strings $(D_X, D_Z)$ which are subsets of the sets $C_X$ and $C_Z$ of check nodes, whose values have been flipped.
This means that the syndromes that are provided as the input of the decoding algorithm are 
\begin{align}
\sigma_X := \sigma_X(E_X) \oplus D_X, \quad \sigma_Z := \sigma_Z(E_Z) \oplus D_Z.
\end{align}

The algorithm we will consider in this work treat $X$ errors and $Z$ errors in a symmetric fashion. More precisely, the decoding algorithm first tries to recover $E_X$ from $\sigma_X$, then proceeds in a similar way to try to recover $E_Z$ from $\sigma_Z$, without exploiting any information about $E_X$ or $\sigma_X$. Said otherwise, the algorithm tries to recover both $E_X$ and $E_Z$ independently. 
For this reason, it will be convenient to restrict our attention to $X$-type errors in the following since $Z$-type error would be treated in the same way. 
In particular, an error will correspond to two sets: a subset $E$ of the qubits and a subset $D$ of the check nodes. 

While considering independent errors is natural in the context of quantum error correction, assuming independence isn't really justified for fault-tolerance since errors will tend to propagate through the circuit for instance, and will therefore likely be correlated. For this reason, the usage is to consider a weaker model of stochastic, but locally decaying noise when studying quantum fault-tolerance: this is the local stochastic error model. 
\begin{definition}[Local stochastic error model]\label{model erreur lc}\
  \\Let $V$ be the set of qubits and $C$ be the set of check nodes. A random error $(E, D)$ with $E \subseteq V$ and $D \subseteq C$ satisfies the local stochastic error model with parameters $(p,q)$ if for all $S \subseteq V$ and $T \subseteq C$, we have
  \begin{align}
  \mathbb{P}[S \subseteq E, T \subseteq D] \leq p^{|S|} q^{|T|}. 
  \end{align}
  If $q=0$, \textit{i.e.}, there are no errors on the syndrome, then we talk of a local stochastic model of parameter $p$.
 In other words, the location of the errors is arbitrary but the probability of a given error decays exponentially with its weight.
\end{definition}

  In this paper, we study a variant of \Cref{algo decodage qec0} that allows us to deal with syndrome errors. This is \Cref{algo decodage qec} below (see \cref{eq:notations1,eq:notations2} for notations). The three differences with \Cref{algo decodage qec0} are the input $\sigma = \sigma_X(E) \oplus D$ instead of $\sigma = \sigma_X(E)$, the while loop condition $\Delta(\sigma_i, F_i) \geq \beta |\sigma_X(F_i)|$ instead of $\Delta(\sigma_i, F_i) > 0$ and the use of $\cF$ instead of $\cF_0$ as set of possible flips (see \Cref{rk:algo} for a discussion about these changes):
  \begin{align}\label{eq:notations2}
    \cF := \left\{F \subseteq \Gamma_Z(g): g \in C_Z \text{ and } |\sigma_X(F)| \geq \frac{d_A}{2} |F|\right\}.
  \end{align}
\begin{algorithm}[H]
  \caption{: Small-set-flip decoding algorithm for quantum expander codes of parameter $\beta \in (0; 1]$
  }\label{algo decodage qec}
    {\bf INPUT:} $\sigma \subseteq {C}_X$ a syndrome such that $\sigma = \sigma_X(E) \oplus D$ for some (unknown) $E \subseteq {V}$ and $D \subseteq C_X$ 
    \\{\bf OUTPUT:} $\hat{E}\subseteq {V}$, a guess for the error pattern (alternatively, a set of qubits to correct)
    \hrule
    \begin{algorithmic}
      \State{$\hat{E}_0 = 0$ ; $\sigma_0 = \sigma$ ; $i = 0$}
      \While{$\displaystyle \exists F_i \in \cF:  \Delta(\sigma_i, F_i) \geq \beta |\sigma_X(F_i)|$}
      \State{$\hat{E}_{i+1} = \hat{E}_i \oplus F_i$}
      \State{$\sigma_{i+1} = \sigma_i \oplus \sigma_X (F_i)$ \qquad // $\sigma_{i+1} = \sigma_X (E \oplus \hat{E}_{i+1}) \oplus D$}
      \State{$i = i+1$}
      \EndWhile
      \State{\Return{$\hat{E}_i$}}
    \end{algorithmic}
\end{algorithm}

\begin{remark}\label{rk:algo}
  In order to simplify our discussion in the paper, we will say that the input of \Cref{algo decodage qec} is $(E,D)$ when its input is $\sigma_X(E) \oplus D$, we will call $\hat{E}$ the output, we will call $E \oplus \hat{E}$ the remaining error, we will denote by $f$ the number of steps and we will call $U = E \cup F_0 \cup \ldots \cup F_{f-1}$ the execution support.
  
  Using $\cF$ instead of $\cF_0$ in \Cref{algo decodage qec} is not restrictive because if the condition $|\sigma_X(F)| \geq \frac{d_A}{2} |F|$ is not satisfied for some $F \subseteq \Gamma_Z(g)$ then this condition is satisfied by $F' = \Gamma_Z(g) \setminus F$ (see the proof of \Cref{existance of a critical generator and error}).

    In \Cref{algo decodage qec0}, the weaker ``while loop condition'' $\Delta(\sigma_i, F_i) > 0$ was used, but it turns out that if $D = \varnothing$ then with high probability on the choice of $E$, the condition $\Delta(\sigma_i, F_i) \geq (1 - 8 \delta) |\sigma_X(F_i)|$ is automatically satisfied at each step of \Cref{algo decodage qec0} (this property was used in the proof of \Cref{main theorem}). On the other hand when $D \neq \varnothing$, requiring $\Delta(\sigma_i, F_i) \geq \beta |\sigma_X(F_i)|$ with $\beta$ close to $1$ makes \Cref{algo decodage qec} more robust against syndrome errors.

\end{remark}

The behavior of \Cref{algo decodage qec} in the particular case where $D = \varnothing$ could be studied following the proof of \cite{fawzi2018efficient}: given $\cQ_G$ a quantum expander code constructed using $G$ some bipartite $(\gamma, \delta)$-expander graph with $\delta < 1/8$, it is possible to prove that a random error $E$ is corrected with high probability by the small-set-flip algorithm of parameter $\beta_0$ ($\beta_0$ as defined in \Cref{subsec:notations}). In the noisy case $D \neq \varnothing$, we cannot hope to entirely correct the error because any single qubit error cannot be distinguished from a well-chosen constant weight syndrome bit error. But we will prove in \Cref{thm:correction} that when $\delta < 1/16$, the correction provided by the small-set-flip algorithm of parameter $\beta < \beta_1$ ($\beta_1$ as defined in \Cref{subsec:notations}) leads to a residual error that is local stochastic with controlled parameter.

\subsection{Fault-tolerant protocol with quantum expander codes}
\label{sec:FT}

The objective of this section is to provide a brief description of Gottesman's framework applied to quantum expander codes. A fault-tolerant protocol is a procedure that transforms an ideal circuit $\mathrm{C}$ into a fault-tolerant one $\mathrm{\tilde{C}}$ in a way that ensures that even if the individual components of $\mathrm{\tilde{C}}$ are noisy, the output of the $\mathrm{\tilde{C}}$ can be used to reproduce the output of $\mathrm{C}$ (in an approximate sense we will specify shortly). 

\paragraph{Definition of a circuit} We start by defining a circuit in this context. A circuit $\mathrm{C}$ can be described by local operations acting on a collection of (at most) $k$ wires, each containing a qubit. More precisely, a circuit of depth $d$ has $d$ time steps and in each time step an operation may be applied to each one of the wires. Each such operation is called a \emph{location} and it is even useful to also consider the operation of not doing anything on a wire as a ``wait'' location (this is used to model memory errors). As such, our circuit $\mathrm{C}$ will have the following types of locations acting on one of the $k$ wires: wait, preparation of $\ket{0}$ (for an inactive wire), Hadamard gate, $R_{\pi/4}$ gate, $R_{\pi/8}$ gate where $R_{\phi} = \begin{pmatrix} e^{-i\phi} & 0 \\ 0 & e^{i \phi} \end{pmatrix}$, a measurement in the computational eigenbasis of $Z$ (the corresponding wire then becomes inactive and the outcome is stored in a classical wire). In order to do useful computations, a two-qubit gate is needed, the controlled-not (CNOT) gate is a location acting on two of the wires. Another item is needed to describe the behavior of the circuit: classical computations on the classical wires of the circuit. As we are going to assume that the classical computation is not subject to faults, we will not describe explicitly the classical computations in terms of a circuit but rather it will be considered as a computation that takes place in between time steps. For this reason, it is important that the classical computations are fast, and this will be the case for us as all the classical computations can be done in parallel in $\cO(\log k)$ time. The location present at any point in the circuit can be controlled by the classical wires. Note that circuit $\mathrm{C}$ does not take any input and all quantum wires are initialized with a state preparation $\ket{0}$, and the output is composed of some classical wires and some quantum wires. As we are most often interested in a classical result, we will assume the circuits ends with a measurement of all the qubits, and hence the output of the circuit is described a random variable $Y \in \{0,1\}^m$, with a probability distribution $P_{Y}$.

\paragraph{Noise model} An ideal circuit will behave exactly as prescribed. We now would like to define a noise model, called \emph{local stochastic model}, for a circuit. The subset of locations that are faulty is a random variable $F$ with a distribution restricted by a parameter $p$. Namely, the set $F$ of faulty locations has an arbitrary distribution that satisfies the property for any subset $S$ of locations $\mathbb{P}[S \subseteq F] \leq p^{|S|}$. Then, the faulty locations are replaced by arbitrary quantum operations with the same input and output space as the original location. We also highlight some important assumptions that are implicit in our definition of circuit. Namely, it is possible to act on all the qubits in parallel. In addition, it is possible to introduce fresh ancilla qubits at any point during the computation, \textit{via} a prepare $\ket{0}$ location. Another  assumption that was already mentioned is that classical computation is not subject to faults and is fast enough so that it can be done between the time steps of the main circuit.

If the noise model we just described is applied to a circuit $\mathrm{C}$, in general, its output will have little to do with the ideal output $P_{Y}$ of $\mathrm{C}$. Our objective is to build a new circuit $\mathrm{\tilde{C}}$ which acts as a robust version of $\mathrm{C}$. In particular, the output of the circuit $\mathrm{\tilde{C}}$ should have as output a random variable $\tilde{Y} \in \{0,1\}^{m'}$ that should behave as an \emph{encoded} version of the ideal outcome $Y$ in some error correcting code. More precisely, we want that there exists a function $\mathrm{D}$ (which should be efficiently computable) such that the distributions $P_{\mathrm{D}(\tilde{Y})}$ and $P_{Y}$ are $\varepsilon$-close in total variation distance. For example, if we want to solve a decision problem and $Y \in \{0,1\}$, then the output $\tilde{Y}$ of the robust circuit will correspond to the ideal result except with probability at most $\varepsilon$.

\paragraph{Circuit transformation overview}
We start with a circuit $\mathrm{C}$ on $k$ qubits with $f(k)$ locations, and we will assume that the circuit $\mathrm{C}$ is sequential, \textit{i.e.}, at any time step, there is only one location which is not a wait location. Note that this can always be done at the cost of increasing the number of locations to at most $k f(k)$. The main error correcting code we use to construct $\mathrm{\tilde{C}}$ is a quantum expander code $\cQ$ with parameters $[[n',k']]$ with $k'$ chosen later and $n' = \frac{k'}{R}$, where $R$ is the rate of the code (in the notation of Theorem 1, $R$ should be chosen to be basically $1/\eta$). The general structure of the circuit $\mathrm{\tilde{C}}$ is as follows. The $k$ qubits of $\mathrm{C}$ are partitioned into blocks, $B_1, \dots, B_{\ell}$ with each block containing at most $k' = \frac{k}{\ell}$ with $\ell = \log^c k$ and $c$ a well-chosen constant (for simplicity, we assume that $k$ is divisible by $\ell$). The circuit $\mathrm{\tilde{C}}$ is going to contain blocks $\tilde{B}_{j}$ each having $n'$ physical qubits and corresponding to the block $B_j$ encoded using the code $\cQ$. The circuit $\mathrm{\tilde{C}}$ is going to alternate between two types of cycles: a simulation cycle and an error correction cycle. In an error correction cycle, we perform a measurement of the syndromes for the blocks $\tilde{B}_j$ for all $j$ in parallel, then use the small-set-flip decoding algorithm to determine the error and apply the corresponding correction to each block. 
In a simulation cycle, we simulate an action of the circuit $\mathrm{C}$. Recall that we assumed that at each time step of $\mathrm{C}$, there is only one location which is not a wait location. The objective will be to simulate this location. Note that in a simulation cycle, only one block $\tilde{B}_j$ will be involved if there are no locations that are across different blocks or two blocks $\tilde{B}_{j_1}$ and $\tilde{B}_{j_2}$ if the location is a CNOT between two qubits, one in the block $B_{j_1}$ and one in $B_{j_2}$. Recall that our objective is to obtain that under the noise model described above, the output is $\varepsilon$-close to an ideal output. For this reason, we define $\varepsilon_0 = \frac{\varepsilon}{f(k)}$ which can be interpreted as an allowed failure rate per location, for a logical operation.

In the following paragraphs, we give a bit more details on how these cycles are performed. We refer the reader to~\cite{gottesman2014fault} for the analysis (and other variants that could be used).

\subsubsection*{Error correction cycle}

For each block, we perform a measurement of the syndrome. For this, we use $n' - k'$ new qubits, one for each generator, prepared using a ``prepare $\ket{0}$'' location. For a given generator (which has a constant weight at most $r$), we start by applying a Hadamard gate, then we apply controlled-X (\textit{i.e.}, CNOT) and controlled-Z (\textit{i.e.}, CNOT conjugated by Hadamard on the target qubit) gates for each non-identity element of the generator, then a Hadamard gate and a measurement. If the gates and measurements are perfect, this is a measurement of the generator. In the presence of noise, one can still bound the propagation of errors because each generator measurement acts on a constant number of qubits and each qubit is involved in a constant number of generators, thanks to the LDPC character of the quantum code. This measurement is performed for each generator of the code $\cQ$. We will perform in parallel measurements for generators that act on disjoint qubits (for example using the partition described in Section~\ref{sec:parallel}). And all the blocks $\tilde{B}_j$ are treated in parallel.

Once we obtain the results of the measurements, we then apply the small-set-flip decoding algorithm (which is a classical algorithm) to determine the error pattern on the physical qubits and then apply the corresponding corrections.

\paragraph{Cost of an error correction cycle} In terms of memory, note that the number of new qubits used for this step is $(n' - k') \ell = (\frac{1}{R} - 1) k$. We remark that we could perform the error correction less often in order to reduce the memory overhead. For example, for an integer $s$, at the $p$-th error correction cycle we could only perform error correction for the blocks $\tilde{B}_j$ where $j = p \mod s$. In this case, the number of qubits used for the measurement is at most $\ceil*{(\frac{1}{R} - 1) \frac{k}{s}}$. This will be at the cost of decreasing the allowed probability of error per location, but for any constant $s$, the allowed error probability per location will remain constant.

In terms of time, the number of time steps (or depth) in the quantum circuit $\mathrm{\tilde{C}}$ is constant. In addition, the classical computation time for the small-set-flip decoding algorithm can be done in time $\cO(\log k)$ (see Section~\ref{sec:parallel}).

\subsubsection*{Simulation cycle}

For a simulation cycle, the structure of the operations we will apply is as follows. We use new qubits to create a well-chosen ancilla state (depending on the type of location we want to simulate), then perform a Bell measurement between the relevant block(s) and the ancilla state and then perform some corrections. In order to construct the well-chosen ancilla state $\ket{\Psi}$, we use the following lemma, which is proved in~\cite{gottesman2014fault}.
\begin{lemma}\label{lem:concat}
  Given $\cD$ a quantum circuit with output given by the $m$-qubit state $\ket{\Psi}$ and any $\delta>0$, there exists a quantum circuit $\cD'$ satisfying the following properties.
  \begin{itemize}
  \item If the circuit undergoes local stochastic noise with sufficiently small parameter $p$, then there exists a failure event that has probability at most $\delta$ and conditioned on the failure event not happening, the output of the circuit is described by $\ket{\Psi}$ to which a local stochastic error is applied.
    \item The number of qubits in the circuit $\cD'$ is at most $c_0 m \log^{c_1} (|\cD|/\delta))$ where $|\cD|$ is the number of locations in $\cD$, for some constants $c_0, c_1$.
  \item When $\cD$ is of depth $d$ and the classical computation takes $t$ time steps, the depth of $\cD'$ is $\cO(d + \log(\log(|\cD|/\delta)))$ and the classical computation time is $\cO(t + \log(\log(|\cD|/\delta)))$.
  \end{itemize}
\end{lemma}

For each time step in $\mathrm{C}$, we simulate a single location which is not a wait location. If this location affects one qubit, we assume for simplicity of notation that it acts on the first qubit of block $B_j$. Note that we always take $\delta = \varepsilon_0 := \frac{\varepsilon}{f(k)}$ when applying Lemma~\ref{lem:concat}. We consider the different types of location case by case.
\begin{itemize}
\item To simulate a preparation of $\ket{0}$, if the block $\tilde{B}_j$ was already created, then there is nothing we need to do. If the block $\tilde{B}_j$ does not exist yet, we will create it prepared in the state $\ket{0}^{\otimes k'}$ encoded in the code $\cQ$. In order to do this, we apply Lemma~\ref{lem:concat} with the circuit $\cD$ being an encoding circuit for the code $\cQ$ initialized in the state $\ket{0}^{\otimes k'}$. 
\item To simulate a Hadamard (an $R_{\pi/4}$ gate is similar), we will apply gate teleportation by preparing an ancilla state as an encoded entangled pair to which $H$ is applied. More precisely, we apply Lemma~\ref{lem:concat} with the circuit $\cD$ being two parallel encoding circuits for the code $\cQ$ with input initialized to $(I \otimes H) \frac{1}{\sqrt{2}} (\ket{00} + \ket{11}) \otimes (\frac{1}{\sqrt{2}} (\ket{00} + \ket{11}))^{\otimes (k'-1)}$. As a result, we obtain a state $\ket{\Psi} \in (\mathbb{C}^2)^{\otimes 2n'}$. We call $D_1$ and $D_2$ the two blocks of $n'$ qubits. We then perform $n'$ Bell measurements in parallel on physical qubits of blocks $\tilde{B}_j$ and $D_1$. This allows us to deduce the outcome of the Bell measurement on the first logical qubit encoded in $\tilde{B}_j$ together with the first logical qubit encoded in $D_1$. Depending on the outcome, a Pauli correction should be on the first logical qubit of $D_2$, which can be done transversally by applying some Pauli operators on each qubit of $D_2$. The block $D_2$ then takes the place of the block $\tilde{B}_j$.
\item To simulate a $R_{\pi/8}$-gate, we similarly apply Lemma~\ref{lem:concat} with the circuit $\cD$ being two parallel encoding circuits for the code $\cQ$ with input initialized to $(I \otimes R_{\pi/8}) \frac{1}{\sqrt{2}} (\ket{00} + \ket{11}) \otimes (\frac{1}{\sqrt{2}} (\ket{00} + \ket{11}))^{\otimes (k'-1)}$. We similarly apply the Bell measurement between the blocks $\tilde{B}_j$ and $D_1$, but the difficulty now is that the possible correction operations are $Z$ and $X R_{\pi/4}$. In order to apply a logical $R_{\pi/4}$ gate, we prepare another pair of ancilla blocks $D_3$ and $D_4$ and use the same procedure as described above to prepare it and apply it.

\item If the location is a CNOT gate between qubits in blocks $B_{j_1}$ and $B_{j_2}$, assuming for simplicity of notation that the CNOT gate is between the first qubit of $B_{j_1}$ and the first qubit of $B_{j_2}$, we will create four blocks $D_1,D_2,D_3,D_4$, each containing $n'$ physical qubits and use Lemma~\ref{lem:concat} to prepare the following state: $(I_{D_1D_2} \otimes \mathrm{CNOT}_{D_3D_4}) \frac{1}{2} (\ket{00}_{D_1D_3} + \ket{11}_{D_1D_3}) \otimes (\ket{00}_{D_2D_4} + \ket{11}_{D_2D_4})\otimes (\frac{1}{\sqrt{2}} (\ket{00} + \ket{11}))^{\otimes (k'-2)}$. Then, as before, we perform a measurement in the Bell basis corresponding to logical qubit 1 of $B_{j_1}$ with logical qubit 1 of $D_1$ and another in the Bell basis corresponding to logical qubit 1 of $B_{j_2}$ with logical qubit 1 of $D_2$ and apply the Pauli corrections to the blocks $D_3$ and $D_4$. The blocks $D_3$ and $D_4$ then play the role of $\tilde{D}_{j_1}$ and $\tilde{D}_{j_2}$. For a CNOT gate between qubits within the same block $B_{j}$, the construction is similar but slightly simpler as we only need two ancillary blocks.
\item To simulate a measurement, we prepare an ancilla state in the encoded $\ket{0}^{\otimes k'}$ state, apply a logical CNOT between the first logical qubit of $\tilde{B}_j$ and the first logical qubit of the ancilla block (using the method described above), then measure all the physical qubits of the ancilla block getting a bit-string of length $n'$, and finally use the small-set-flip decoding algorithm on the $Z$-type generators (this corresponds to a classical parity check matrix) on the classical outcome to get an encoded bit-string. The first logical bit encoded in this bit-string corresponds to the outcome of the measurement we are simulating.
\end{itemize}

\paragraph{Cost of a simulation cycle} In terms of memory, for all the location types, we prepare a state $\ket{\Psi}$ on $\cO(n')$ qubits, and so $\cO(n' \log^{c_1} (n'/\varepsilon_0))$ qubits. By the choice of $k'$, this quantity is sub-linear in the number of qubits $k$. 

In terms of time, the number of time steps added to the quantum circuit $\mathrm{\tilde{C}}$ is $\cO(\log \log f(k))$ and the classical time is $\cO(\log k)$.

\section{Analysis of \Cref{algo decodage qec}}
\label{sec:analysis}

\subsection{Notations}
\label{subsec:notations}

The algorithms of \Cref{sec:analysis} and \Cref{sec:parallel} depend on three parameters $\delta, \beta \in (0;1)$ and $c \in \bR^*_+$. Note that the parameter $c$ and the constants $\chi, c_3, \eta$ and $f_0$ defined below are used in \Cref{sec:parallel} for the analysis of the parallel version of the algorithm but they are not used in \Cref{sec:analysis}.
\\We consider $G = (A \cup B, \cE)$ a $(d_A, d_B)$-biregular $(\gamma, \delta)$-expander graph with $\gamma > 0$, we denote by $\cQ$ the quantum expander code associated to $G$ (see \Cref{subs:qec}) and by $V, C_X, C_Z$ and $n := |V|$ respectively the set of qubits, the set of $Z$-type stabilizer generators, the set of $X$-type stabilizer generators and the number of physical qubits of $\cQ$. We will also use $\Gamma_X$ and $\Gamma_Z$ the neighborhoods in the graphs $G_X$ ans $G_Z$ as defined in \Cref{subs:qec}.
\\We run the small-set-flip decoding algorithm (\Cref{algo decodage qec}) of parameter $\beta$ on input $(E,D)$ where $E \subseteq V$ represents a qubit error and $D \subseteq C_X$ represents a syndrome error, we denote by $\hat{E}$ the output of the algorithm, by $f$ the number of steps and by $U = E \cup F_0 \cup \ldots \cup F_{f-1}$ the execution support.

\smallskip

We also define the constants:
\begin{align*}
  r := d_A/d_B,
  && \gamma_0 = \frac{r^2}{\sqrt{1+r^2}} \gamma,
  && \beta_0 = \beta_0(\delta) := 1 - 8 \delta,
  && \beta_1 = \beta_1(\delta) := 1 - 16 \delta,
\end{align*}
\begin{align*}
  c_0 = c_0(\delta, \beta) := \frac{4}{d_A (\beta_1 - \beta)}, 
  && c_1 = c_1(\delta, \beta) := \frac{\beta_1 - \beta}{\beta_0(1 - \beta)},
\end{align*}
\begin{align*}
  c_2 = c_2(\delta, \beta) := \frac{2 \beta_0}{\beta_1 - \beta},
  && c_3 = c_3(\beta, c) := \frac{2(1 + c)}{\beta d_A},
\end{align*}
\begin{align*}
  \chi := (d_B(d_A - 1) + 1)(d_A(d_B - 1) + 1),
\end{align*}
\begin{align*}
  \alpha_0 = \alpha_0(\beta) := \frac{r \beta}{4 + 2 r \beta},
  && \eta = \eta(\delta, \beta, c) := 1 - \frac{\beta c_1 (c - 1 - c_2)}{d_A d_B \chi c}.
\end{align*}
We also define the function ($s$  will represents the size of the input syndrome for the parallel algorithm \Cref{algo decodage qec parallel}):
\begin{align*}
  f_0 = f_0(\delta, \beta, c) : s \in \bN \mapsto \left\lceil \chi \log_{1/\eta}(s) \right\rceil.
\end{align*}
Note that:
\begin{itemize}
\item If $\delta < 1/8$ then $\beta_0 > 0$.
\item If $\delta < 1/16$ then $\beta_1 > 0$.
\item If $\delta < 1/16$ and $0 < \beta < \beta_1$ then $c_0, c_1, c_2, c_3 > 0$ and $0 < \alpha_0 \leq 1$.
\item If $\delta < 1/16$, $0 < \beta < \beta_1$ and $c > c_2 + 1$ then $\eta < 1$ and $f_0(s) = \Theta(\log(s))$.
\end{itemize}

\subsection{Statements of the theorems}

As discussed previously, the noiseless case $D = \varnothing$ was studied in \cite{fawzi2018efficient}:

\begin{theorem}[\cite{fawzi2018efficient}]\label{thm:stoc}
  We use the notations of \Cref{subsec:notations} with $\delta < 1/8$.
  \\With probability at least $1 - \reducedProba$ on the choice of the local stochastic error $E$, the small-set-flip algorithm (\Cref{algo decodage qec}) with parameter $\beta_0$ outputs some error $\hat{E}$ equivalent to $E$ under the assumption $D = \varnothing$.
\end{theorem}

In this section, we are going to prove \Cref{thm:correction}, a generalized version of \Cref{thm:stoc} that we can apply in the case where the syndrome error $D \subseteq C_X$ is not empty.

\begin{theorem}\label{thm:correction}
  We use the notations of \Cref{subsec:notations} with $\delta < 1/16$ and $\beta < \beta_1$.
  \\There exist constants $p_0 > 0, p_1 > 0$ such that the following holds. Suppose the pair $(E, D)$ satisfies a local stochastic noise model with parameter $(p_{\mathrm{phys}}, p_{\mathrm{synd}})$ where $p_{\mathrm{phys}} < p_0$ and $p_{\mathrm{synd}} < p_1$. Then there exists an event $\mathsf{succ}$ that has probability $1 - \reducedProba$ and a random variable $E_{\textrm{ls}}$ that is equivalent to $E \oplus \hat{E}$ such that conditioned on $\mathsf{succ}$, $E_{\textrm{ls}}$ has a local stochastic distribution with parameter $K p_{\mathrm{synd}}^{1/c_0}$ where $K$ is a constant independent of $p_{\mathrm{synd}}$.
\end{theorem}

\subsection{Small adversarial errors}
\label{subsec: small}

The first step to prove \Cref{thm:correction} is to study the case where the qubit error $E$ can be adversarial but where $E \oplus\hat{E}$ is supposed to be reduced with $|E \oplus\hat{E}| \leq \gamma_0 \sqrt{n}$. Here ``reduced'' means that $E \oplus\hat{E}$ has the smallest Hamming weight among all errors equivalent to $E$. The result in that case is summarized in \Cref{coro:error_qubits_syn}: it is possible to use expansion-based arguments to find an upper bound on $|E \oplus\hat{E}|$ which grows linearly with $|D \cap \sigma_X(E \oplus\hat{E})|$. \Cref{coro:error_qubits_syn} is a consequence of \Cref{prop:error_qubits_syn} and \Cref{lem:robust} that we state now but only prove at the end of this section.

\begin{proposition}\label{prop:error_qubits_syn}
  We use the notations of \Cref{subsec:notations} with $\delta < 1/16$ and $\beta \in (0; \beta_1)$.
  \\If $|E| \leq \gamma_0 \sqrt{n}$ and $|\sigma_X(E)| > c_2 |D \cap \sigma_X(E)|$ then there exists at least one valid $F \in \cF$ for \Cref{algo decodage qec} with parameter $\beta$.
  \\More precisely, let $\sigma = \sigma_X(E) \oplus D$ then the set $G := \left\{F \in \cF: \Delta(\sigma, F) \geq \beta |\sigma_X(F)|\right\}$ satisfies:
  \begin{align*}
    \sum_{F \in G} |\sigma_X(F)|
    \geq c_1 \left[\left|\sigma_X(E)\right| - c_2 |D \cap \sigma_X(E)|\right]
    > 0.
  \end{align*}
\end{proposition}

\begin{lemma}[Robustness]\label{lem:robust}
  We use the notations of \Cref{subsec:notations} with $\delta < 1/8$.
  \\If $E_\textrm{R} \subseteq V$ is a reduced error with $|E_\textrm{R}| \leq \gamma_0 \sqrt{n}$ then:
  \begin{align*}
    |\sigma_X(E_\textrm{R})|
    \geq \frac{\beta_0 d_A}{2} |E_\textrm{R}|.
  \end{align*}
\end{lemma}

Together, \Cref{prop:error_qubits_syn} and \Cref{lem:robust} imply the following:
\begin{corollaire}\label{coro:error_qubits_syn}
  We use the notations of \Cref{subsec:notations} with $\delta < 1/16$ and $\beta \in (0; \beta_1)$.
  \\Suppose that $E \oplus \hat{E}$ is reduced with $|E \oplus \hat{E}| \leq \gamma_0 \sqrt{n}$ then
  \begin{align*}
    |E \oplus \hat{E}|
    \leq c_0 |D \cap \sigma_X(E \oplus \hat{E})|.
  \end{align*}
\end{corollaire}

\begin{proof}
  Using the notations $\sigma_i$ from the body of \Cref{algo decodage qec}, the value of the syndrome $\sigma_f$ at the end of the algorithm is $\sigma_f = \sigma_X(E \oplus \hat{E}) \oplus D$. Since the while loop condition is not satisfied for $\sigma_f$, the contraposition of \Cref{prop:error_qubits_syn} ensures that $|\sigma_X(E \oplus \hat{E})| \leq c_2 |D \cap \sigma_X(E \oplus \hat{E})|$. By \Cref{lem:robust}, $|\sigma_X(E \oplus \hat{E})| \geq \frac{\beta_0 d_A}{2} |E \oplus \hat{E}|$ which concludes the proof.
\end{proof}

The rest of \Cref{subsec: small} is devoted to prove \Cref{prop:error_qubits_syn} and \Cref{lem:robust}.
\\We will study the sets $F \in G$ ($G$ is defined in \Cref{prop:error_qubits_syn}) which would have been flipped during the small-set-flip algorithm with input $E$ and $D = \varnothing$, \textit{i.e.}, without syndrome error.
For a  given set $E \subseteq V = A^2 \uplus B^2$ (where $ \uplus$ stands for disjoint union), we introduce a notation for a normalized Hamming weight:
\begin{align*}
  \|E\| := \frac{|E \cap A^2|}{d_B} + \frac{|E \cap B^2|}{d_A}.
\end{align*}
$\|\cdot\|$ shares a couple of properties with the usual the cardinality $|\cdot|$. In particular it is straightforward to check that for $E, E_1, E_2 \subseteq V$:
\begin{align*}
  & \|E\| = 0 \Leftrightarrow E = \varnothing,
  && d_A \|E\| \leq |E| \leq d_B \|E\|,\\
  & |\sigma_X(E)| \leq d_A d_B \|E\|,
  && \|E_1 \cup E_2\| \leq \|E_1\| + \|E_2\|,\\
  & \|E_1 \uplus E_2\| = \|E_1\| + \|E_2\|,
  && \|E_1 \oplus E_2\| = \|E_1\| + \|E_2\| - 2 \|E_1 \cap E_2\|.
\end{align*}
We will say that a qubit error $E \subseteq V$ is $\|\ \|$-reduced when $\|E\|$ is minimal over $E + {\cC_Z}^{\bot}$.
All along this section we will use the handy property of \Cref{lem:subset reduced}:
\begin{lemma}\label{lem:subset reduced}
  Let $E_1 \subseteq E_2 \subseteq V$ be two errors. If $E_2$ is reduced (resp. $\|\ \|$-reduced) then $E_1$ is reduced (resp. $\|\ \|$-reduced). 
\end{lemma}
\begin{proof}
  Let's prove that for any $E \in \cC_Z^{\perp}, |E_1| \leq |E_1 \oplus E|$. Using the fact that $| \cdot |$ is additive for disjoint unions and non-negative, we have
  \begin{align*}
    0
    \leq |E_2 \oplus E| - |E_2| 
    = | E - E_2 | + |E_2 - E| - | E_2 \cap E | - |E_2 - E| 
    &= |E| - 2 |E_2 \cap E| \\
    \leq |E| - 2 |E_1 \cap E|
    = |E_1 \oplus E| - |E_1|.
  \end{align*}
  As $\| \cdot\|$ is also additive for disjoint unions and non-negative, the same proof also works when replacing $\|\cdot \|$ with $|\cdot |$.
\end{proof}

First of all, we need to study the case where the syndrome is noiseless ($D = \varnothing$). In that case and when the initial graph $G$ is sufficiently expanding, there exists at least one $X$-type stabilizer generator called a ``critical generator'' (this notion was introduced in \cite{leverrier2015quantum}) whose support contains some $F \in \cF$ that decreases the syndrome weight when flipped.

\begin{lemma}[Lemma 8 of \cite{leverrier2015quantum} revisited]\label{existance of a critical generator and error}
  Let $E \subseteq V$ be a $\|\ \|$-reduced error such that $0 < \|E\| \leq \gamma_0 \sqrt{n}/d_A$, then there exists $F \in \cF$ with $F \subseteq E$ and:
  \begin{enumerate}[label=(\roman*)]
  \item  \label{existance of a critical generator and error i}
    $|\sigma_X(F)| \geq \frac{1}{2} d_A d_B \|F\|$,
  \item  \label{existance of a critical generator and error ii}
    $\Delta(\sigma_X(E), F) \geq |\sigma_X(F)| - 4 \delta d_A d_B \|F\|$.
  \end{enumerate}
\end{lemma}

\begin{proof}
  We have:
  \begin{align*}
    |E| \leq d_B \|E\|
    \leq \gamma_0 \sqrt{n} / r
    = \gamma \frac{r}{\sqrt{1 + r^2}} \sqrt{n}
    = \gamma \frac{r}{\sqrt{1 + r^2}} \sqrt{n_A^2 + n_B^2}
    = \gamma n_B
    = \min(\gamma n_A, \gamma n_B).
  \end{align*}
  Since $|E| \leq \min(\gamma n_A, \gamma n_B)$, Lemma 7 of \cite{leverrier2015quantum} states that there exists $g \in C_Z$ called a ``critical generator'' whose neighborhood satisfies $\Gamma_Z(g) = x_a \uplus \overline{x}_a \uplus \chi_a \uplus x_b \uplus \overline{x}_b \uplus \chi_b$ with:
  \begin{align*}
    & x_a \subseteq A^2 \cap E,
    && x_b \subseteq B^2 \cap E,
    && x_a \cup x_b \neq \varnothing,\\
    & \overline{x}_a \subseteq A^2 \setminus E,
    && \overline{x}_b \subseteq B^2 \setminus E,\\
    & \chi_a \subseteq A^2,
    && \chi_b \subseteq B^2,
    && \|\chi_a\| \leq 2 \delta,
    && \|\chi_b\| \leq 2 \delta,
  \end{align*}
  and such that the $Z$-type generators of $C_X$ involving qubits from $x_a \uplus \overline{x}_a \uplus x_b \uplus \overline{x}_b$ do not involve other qubits from $E$. Formally, for $v_a \in x_a, v_b \in x_b, \overline{v}_a \in \overline{x}_a$ and $\overline{v}_b \in \overline{x}_b$:
  \begin{align*}
    & E \cap \Gamma_X[\Gamma_X(v_a) \cap \Gamma_X(v_b)] = \{v_a, v_b\},
    && E \cap \Gamma_X[\Gamma_X(\overline{v}_a) \cap \Gamma_X(\overline{v}_b)] = \varnothing, &&\\
    & E \cap \Gamma_X[\Gamma_X(v_a) \cap \Gamma_X(\overline{v}_b)] = \{v_a\},
    && E \cap \Gamma_X[\Gamma_X(\overline{v}_a) \cap \Gamma_X(v_b)] = \{v_b\}.
  \end{align*}
  Take $F = x_a \uplus x_b$ with the purpose to get $F \subseteq E$. We have:
  \begin{align*}
    |\sigma_X(F)|
    & = |x_a|(|\overline{x}_b| + |\chi_b|) + |x_b|(|\overline{x}_a| + |\chi_a|)\\ &
    = d_A d_B \mathlarger{(}\|x_a\|(\|\overline{x}_b\| + \|\chi_b\|) + \|x_b\|(\|\overline{x}_a\| + \|\chi_a\|)\mathlarger{)}\\ &
    = d_A d_B \mathlarger{(}\|x_a\|(1 - \|x_b\|) + \|x_b\|(1 - \|x_a\|)\mathlarger{)}.
  \end{align*}
  But $\|F\| = \|x_a\| + \|x_b\|$, thus:
  \begin{align}\label{eq:func analysis}
    \frac{|\sigma_X(F)|}{d_A d_B \|F\|} = 1 - \frac{2 \|x_a\|\|x_b\|}{\|x_a\| + \|x_b\|}.
  \end{align}
  Note that we have $0 \leq \|x_a\|,\|x_b\| \leq 1$ by definition, we have $0 < \|x_a\| + \|x_b\|$ because $x_a \cup x_b \neq \varnothing$ and we have $\|x_a\| + \|x_b\| \leq 1$ because $E$ is $\|\ \|$-reduced. With these constraints on $\|x_a\|$ and $\|x_b\|$, a function analysis of the right hand side of \cref{eq:func analysis} gives \cref{existance of a critical generator and error i} and $F \in \cF$.
  \\For \cref{existance of a critical generator and error ii}, we lower bound $\Delta(\sigma_X(E), F)$ using the inequalities $\|\chi_a\| \leq 2 \delta$ and $\|\chi_b\| \leq 2 \delta$:
  \begin{align*}
    \Delta(\sigma_X(E), F)
    & \geq |x_a||\overline{x}_b| + |x_b||\overline{x}_a| - |x_a||\chi_b| - |x_b||\chi_a|\\
    & = |\sigma_X(F)| - 2 d_A d_B (\|x_a\|\|\chi_b\| + \|x_b\|\|\chi_a\|)\\
    & \geq |\sigma_X(F)| - 4 \delta d_A d_B (\|x_a\| + \|x_b\|)\\
    & = |\sigma_X(F)| - 4 \delta d_A d_B \|F\|.
  \end{align*}
\end{proof}

\begin{proof}[Proof of \Cref{prop:error_qubits_syn} and \Cref{lem:robust}.]

  Both proofs begin in the same way: we set $E_0$ to be the $\|\ \|$-reduced error equivalent to $E$ (or equivalent to $E_\textrm{R}$ in the case of \Cref{lem:robust}), we apply \Cref{existance of a critical generator and error} to $E_0$ which provides some $F_0 \subseteq E_0$ and we define $E_1 := E_0 \oplus F_0 = E_0 \setminus F_0$. More generally, we set by induction $E_{i+1} = E_i \oplus F_i = E_i \setminus F_i$ where $F_i$ is obtained by applying \Cref{existance of a critical generator and error} to $E_i$. This construction is licit (\textit{i.e.}, we can apply \Cref{existance of a critical generator and error} to $E_i$) because $E_i$ is $\|\ \|$-reduced as a subset of the $\|\ \|$-reduced error $E_0$ (see \Cref{lem:subset reduced}), and $\|E_i\| \leq \|E_0\| \leq \|E\| \leq |E|/d_A \leq \gamma_0 \sqrt{n}/d_A$. Let $f'$ be the last step of this procedure then $\|E_{f'}\| = 0$ and thus:
  \begin{align}\label{eq:E eq uplus F}
    E_0 = \biguplus_{i=0}^{f'-1} F_i.
  \end{align}
  Since the sets $F_i$ are those given in \Cref{existance of a critical generator and error}, we can use \Cref{existance of a critical generator and error} \cref{existance of a critical generator and error ii,existance of a critical generator and error i} to lower bound $|\sigma_X(E_0)|$:
  \begin{align}\label{eq:rewrite sigma}
    |\sigma_X(E_0)|
    = \sum_{i = 0}^{f'-1}  \Delta(\sigma_X(E_i), F_i)
    \geq \sum_{i=0}^{f'-1} |\sigma_X(F_i)| - \sum_{i=0}^{f'-1} 4 \delta d_A d_B \|F_i\|
    \geq \frac{\beta_0 d_A d_B}{2} \sum_{i = 0}^{f'-1} \|F_i\|.
  \end{align}
  Because $E_0 = \biguplus_i F_i$ (\cref{eq:E eq uplus F}), we have:
  \begin{align}\label{eq:robust}
    |\sigma_X(E_0)| \geq \frac{\beta_0 d_A d_B}{2} \|E_0\|.
  \end{align}
  \\The above arguments hold for \Cref{prop:error_qubits_syn} as well as for \Cref{lem:robust}. Proving \Cref{lem:robust} is now direct:
  \begin{align*}
    |\sigma_X(E_\textrm{R})|
    = |\sigma_X(E_0)|
    \geq \frac{\beta_0 d_A d_B}{2} \|E_0\|
    \geq \frac{\beta_0 d_A}{2} |E_0|
    \geq \frac{\beta_0 d_A}{2} |E_\textrm{R}|.
  \end{align*}

  For \Cref{prop:error_qubits_syn}, let us start by providing an overview of how we will proceed. Note that a union bound yields $|\sigma_X(E)| = |\sigma_X(E_0)| \leq \sum_{i=0}^{f'-1} |\sigma_X(F_i)|$. In fact,  we will prove in \cref{eq:disjoint} below that that this upper bound is nearly tight:  $|\sigma_X(E)| \geq \beta_0 \sum_{i=0}^{f'-1} |\sigma_X(F_i)|$ and $\beta_0$ is arbitrarily close to $1$ when $\delta$ is small. Intuitively, this means that the intersection of the sets $\sigma_X(F_i)$ is small and thus $|\sigma_X(E) \cap \sigma_X(F_i)|$ is generally large. This is still true if the size of the syndrome error $D$ is small, \textit{i.e.}, it holds that $|\sigma \cap \sigma_X(F_i)|$ is generally large. Hence we will obtain \Cref{prop:error_qubits_syn} by computing the average of the quantity $|\sigma \cap \sigma_X(F_i)|$ over the sets $F_i$. We now provide the details:
  \begin{align*}
    |\sigma_X(E)|
    = |\sigma_X(E_0)|
    & \geq \sum_{i=0}^{f'-1} |\sigma_X(F_i)| - \sum_{i=0}^{f'-1} 4 \delta d_A d_B \|F_i\| && \text{by \cref{eq:rewrite sigma}}\\
    & = \sum_{i=0}^{f'-1} |\sigma_X(F_i)| - 4 \delta d_A d_B \|E_0\| && \text{by \cref{eq:E eq uplus F}}\\
    & \geq \sum_{i=0}^{f'-1} |\sigma_X(F_i)| - \frac{8 \delta}{\beta_0} |\sigma_X(E)| && \text{by \cref{eq:robust}}.
  \end{align*}
  Hence we have:
  \begin{align}\label{eq:disjoint}
    \left(1 + \frac{8 \delta}{\beta_0}\right) |\sigma_X(E)|
    \geq \sum_{i=0}^{f'-1} |\sigma_X(F_i)|.
  \end{align}
  The relation between $\Delta(\sigma, F)$ and $\left|\sigma \cap \sigma_X(F) \right|$ is given by:
  \begin{align*}
    \Delta(\sigma, F)
    = |\sigma| - |\sigma \oplus \sigma_X(F)|
    = 2 \left|\sigma \cap \sigma_X(F) \right| - |\sigma_X(F)|
  \end{align*}
  where we have used the equality $|A_1 \oplus A_2| = |A_1| + |A_2| - 2|A_1 \cap A_2|$.
  \\In particular, when $F \notin G$:
  \begin{align}\label{ineg2}
    \left|\sigma \cap \sigma_X(F) \right|
    \leq \frac{1 + \beta}{2} |\sigma_X(F)|.
  \end{align}
  \\On the one hand, \cref{ineg2,eq:disjoint} give an upper bound on the sum $S := \sum_{i=0}^{f'-1} \left|\sigma \cap \sigma_X(F_i) \right|$:
  \begin{align*}
    S
    & = \sum_{F_i \in G} \left|\sigma \cap \sigma_X(F_i) \right| + \sum_{F_i \notin G} \left|\sigma \cap \sigma_X(F_i) \right|\\
    & \leq \sum_{F_i \in G} |\sigma_X(F_i)| + \frac{1 + \beta}{2} \sum_{F_i \notin G} |\sigma_X(F_i)| && \text{by \cref{ineg2}}\\
    & = \frac{1 - \beta}{2} \sum_{F_i \in G} |\sigma_X(F_i)| + \frac{1 + \beta}{2} \sum_{i=0}^{f'-1} |\sigma_X(F_i)|\\
    & \leq \frac{1 - \beta}{2} \sum_{F_i \in G} |\sigma_X(F_i)| + \frac{1 + \beta}{2\beta_0} |\sigma_X(E)| && \text{by \cref{eq:disjoint}.}
  \end{align*}
  On the other hand, $E_0 = \biguplus_i F_i$ (\cref{eq:E eq uplus F}) implies that $\sigma_X(E) = \sigma_X(E_0) = \bigoplus_{i = 0}^{f'-1} \sigma_X(F_i)$ and thus $S$ is lower bounded by:
  \begin{align*}
    S
    & \geq \left|\sigma \cap \sigma_X(E) \right|\\
    & = \left|(\sigma_X(E) \oplus D) \cap \sigma_X(E)\right|\\
    & = \left|\sigma_X(E) \oplus (D \cap \sigma_X(E))\right|\\
    & = \left|\sigma_X(E)\right| - |D \cap \sigma_X(E)|.
  \end{align*}
  Combining both inequalities we get \Cref{prop:error_qubits_syn}.
\end{proof}

\subsection{Random errors of linear size}
\label{subsec:cst}

The upper bound given by \Cref{prop:error_qubits_syn} can be applied for qubit errors of size up to $t = \cO(\sqrt{n})$. In the case of a local stochastic noise, the errors have a typical size of $\Theta(n)$. The relationship between the two frameworks is given by percolation arguments close to the arguments used in \cite{fawzi2018efficient}.

Percolation arguments will allow us to decompose a random error $E$ as a disjoint union of small error sets, each of which has size upper bounded by $t$. The study of these small errors has been done in \Cref{subsec: small} and a consequence of \Cref{coro:error_qubits_syn} is that when we use \Cref{algo decodage qec} to correct it, the remaining error is local stochastic. Moreover, \Cref{algo decodage qec} is intuitively local in the sense that two qubit errors far away in the factor graph of the code will not interact during the decoding procedure. Under the assumption that a local stochastic error is produced for each small error set, the locality property will allow us to conclude that when we correct the initial error $E$ with \Cref{algo decodage qec}, the remaining error has a local stochastic distribution.
\\In order to formalize the notion of locality, we define $\cG$ called the syndrome adjacency graph of the code in the following way: $\cG$ is equal to $G_X$ (as defined in \Cref{subs:qec}) with additional edges between the qubits which share an $X$-type or a $Z$-type generator. In other words, the set of vertices of $\cG$ is indexed by $\cV := V \cup C_X$ the set of qubits and the set of $Z$-type generators, a $Z$-type generator is incident to the qubits in its support and two qubits are linked when they are both in the support of the same generator. Note that the degree of $\cG$ is upper bounded by $d := \maxDeg$. Using the execution support $U = E \cup F_0 \cup \ldots \cup F_{f-1}$, it is easy to decompose the error into small error sets: each connected component $K$ of $U$ provides one error set $K \cap E$.
\begin{lemma}[Locality of \Cref{algo decodage qec}, \cite{fawzi2018efficient}]\label{lem:local}
  We use the notations of \Cref{subsec:notations}.
  \\For any set $K \subseteq V$ with $\Gamma(K) \cap \Gamma(U \setminus K) = \varnothing$ in $\cG$, there is a valid execution of \Cref{algo decodage qec} on the input $(E \cap K, D \cap \Gamma_X(K))$ whose output is $\hat{E} \cap K$ and whose support is $U \cap K$.
\end{lemma}
What is the size of the remaining error $(E \oplus \hat{E}) \cap K$? We will show that this size is small enough to apply \Cref{coro:error_qubits_syn}. The key point is to note that among the vertices of $X := K \cup (D \cap \Gamma_X(K)) \subseteq \cV$, there is at least a fraction $2 \alpha_0$ of these vertices which belong to $E \cup D$ ($\alpha_0$ is defined in \Cref{subsec:notations}). We will say that $X$ is a $2 \alpha_0$-subset of $E \cup D$ (see \Cref{def:alpha}) and percolation arguments (see \Cref{lem:perco}) will show that with high probability, any connected $\alpha$-subset of a random error $E \cup D$ must be small enough to apply \Cref{coro:error_qubits_syn}.
\begin{definition}[\cite{fawzi2018efficient}]\label{def:alpha}
  Let $\cG = (\cV, \cE)$ be a graph, let $\alpha \in (0;1]$ and let $E,X \subseteq \cV$. $X$ is said to be an \emph{$\alpha$-subset of $E$} if $|X \cap E| \geq \alpha |X|$.
    We also define the integer $\textrm{MaxConn}_{\alpha}(E)$ by:
    $$
    \textrm{MaxConn}_{\alpha}(E) = \max \{ |X|: X \text{ is a connected in $\cG$ and is an $\alpha$-subset of $E$}\}
    $$
\end{definition}
 This notion of $\alpha$-subset is relevant because if we run the small-set-flip decoding algorithm of parameter $\beta > 0$ and set $U = E \cup F_0 \cup \ldots \cup F_{f-1}$ to be the execution support then $U \cup D$ is a $2 \alpha_0$-subset of $E \cup D$.

  Later we will need the following technical lemma in order to reduce to the case where the remaining error $E \oplus \hat{E}$ is reduced:
\begin{lemma}\label{lem:cup alpha subset}
  Let $E, X_1, X_2 \subseteq \cV$ with $|X_2| \leq |X_1|$. If $X_1$ is an $\alpha$-subset of $E$ then $X_1 \cup X_2$ is an $\frac{\alpha}{2}$-subset of $E$.
\end{lemma}
\begin{proof}
  By assumption $|X_1| \leq \frac{1}{\alpha} |X_1 \cap E|$ thus:
  \begin{align*}
    |X_1 \cup X_2|
    \leq 2 |X_1|
    \leq \frac{2}{\alpha} |X_1 \cap E|
    \leq \frac{2}{\alpha} |(X_1 \cup X_2) \cap E|.
  \end{align*}
\end{proof}

\begin{proposition}\label{linear number of flips synd error}
  Using the notations of \Cref{subsec:notations}, $U \cup D$ is a $2 \alpha_0$-subset of $E \cup D$.
\end{proposition}
\begin{proof}
  Using the notations $\sigma_i$ and $F_i$ from the body of \Cref{algo decodage qec}, we have $|\sigma_i| - |\sigma_{i+1}| \geq \beta |\sigma_X(F_i)| \geq \frac{\beta d_A}{2} |F_i|$ thus we can lower bound $|\sigma_X(E) \oplus D| = |\sigma_0|$ by:
  \begin{align*}
    |\sigma_X(E) \oplus D|
    \geq |\sigma_0| - |\sigma_f|
    = \sum_{i=0}^{f-1} |\sigma_i| - |\sigma_{i+1}|
    \geq \frac{\beta d_A}{2} \sum_{i=0}^{f-1} |F_i|.
  \end{align*}
  But $|\sigma_X(E) \oplus D| \leq d_B |E| + |D| \leq d_B |E \cup D|$ thus $\sum_i |F_i| \leq \frac{2}{\beta r} |E \cup D|$ and:
  \begin{align*}
    \left|U \cup D\right|
    \leq |E \cup D| + \sum_{i=0}^{f-1} |F_i|
    \leq \frac{1}{2 \alpha_0} |E \cup D|
    = \frac{1}{2 \alpha_0} |(U \cup D) \cap (E \cup D)|
    .
  \end{align*}
\end{proof}

In order to prove \Cref{thm:correction}, we would like to show that with probability at least $1 - \reducedProba$, there is a reduced error $E_\textrm{ls}$ equivalent to $E \oplus \hat{E}$ which is local stochastic.

  Recall from \Cref{model erreur lc} that an error $E_\textrm{ls}$ is local stochastic with parameter $p$ if we can upper bound the probability $\mathbb{P}[S \subseteq E_\textrm{ls}]$ by $p^{|S|}$ for all $S \subseteq V$. The reason why the probability $p$ provided by \Cref{thm:correction} depends on $p_{\mathrm{synd}}$ is that we will use the hypothesis that $D$ is local stochastic: $\mathbb{P}[T \subseteq D] \leq p_{\mathrm{synd}}^{|T|}$ for all $T \subseteq C_X$. In order to establish \Cref{thm:correction}, it would be sufficient to prove that for all $S \subseteq E_\textrm{ls}$, the size of $D \cap \Gamma_X(S)$ is lower bounded by a linear function of $|S|$. The particular case where $S = E_\textrm{ls} = E \oplus \hat{E}$ has already been proven in \Cref{coro:error_qubits_syn} but unfortunately it is not possible to prove this property for all $S \subseteq E_\textrm{ls}$. Indeed, for an initial error $E = \varnothing$ and $S = \hat{E} \setminus \Gamma_X(D)$, such a property would imply that $S = \varnothing$, \textit{i.e.}, $\hat{E} \subseteq \Gamma_X(D)$. The last statement is clearly false because on input $(E, D)$ with $E = \varnothing$, \Cref{algo decodage qec} could flip some qubit which is not in $\Gamma_X(D)$ (note, however, that such a flip would necessarily happen after several steps).

Therefore, instead of lower bounding $|D \cap \Gamma_X(S)|$ linearly in $|S|$, we will lower bound $|D \cap \Gamma_X(W)|$ linearly in $|W|$ where $W$ contains $S$ and is such that we can apply the locality property of \Cref{lem:local} to prove that \Cref{algo decodage qec} cannot flip any set on input $(W, D \cap \Gamma_X(W))$. Applying \Cref{coro:error_qubits_syn} for the execution of \Cref{algo decodage qec} on input $(W, D \cap \Gamma_X(W))$ will provide the desired lower bound. We will call such a $W$ a \emph{witness} for $S$ (see \Cref{def:witness}).

  In the proof of \Cref{thm:correction}, we will show that if we have witnesses for any $S \subseteq E_\textrm{ls}$ then $E_\textrm{ls}$ is local stochastic. This will require to be able to control the number of possible witnesses for a given $S$ and that is why we will need an additional condition $W \in \mathcal{M}(S)$ in the definition of a witness (see \Cref{def:M(S)} for $\mathcal{M}(S)$ and \Cref{def:witness} for witness).

\begin{definition}[\cite{gottesman2014fault}]\label{def:M(S)}
  Let $\cG = (V \cup C_X, \cE)$ be the syndrome adjacency graph and $d :=  \maxDeg$ an upper bound on the degrees in $\cG$ (as defined above \Cref{lem:local}). For $S \subseteq V$, we define $\mathcal{M}(S) \subseteq \cP(V)$ as the set of all subsets $W \subseteq V$ such that $S \subseteq W$ and such that any connected component $W'$ of $W$ in $\cG$ satisfies $W' \cap S \neq \varnothing$.
 \end{definition}
In this definition, a set $W \in \mathcal{M}(S)$ is any superset of $S$ such that $S$ intersects every connected component of $W$.
Lemma 2 of Ref.~\cite{gottesman2014fault} provides an upper bound on the number of sets $W \in \mathcal{M}(S)$ of a given size:
  \begin{align}\label{Ineq M gottesman}
    \left|\left\{W \in \mathcal{M}(S): |W| = t\right\}\right| \leq \frac{(ed)^t}{ed^{|S|}}.
  \end{align}

\begin{definition}\label{def:witness}\ 
  Let $D \subseteq C_X$ be a syndrome error and $c$ be a constant. For $S \subseteq V$, we say that $W \subseteq V$ is a \emph{$c$-witness for $(S,D)$} if $W \in \mathcal{M}(S)$ and $|W| \leq c |D \cap \Gamma_X(W)|$.
\end{definition}

  The proofs are organized as follow: first we show that we can easily find a witness for any $S \subseteq E \oplus \hat{E}$ under the assumptions that $|E \oplus \hat{E}| \leq \gamma_0 \sqrt{n}$ and that $E \oplus \hat{E}$ is reduced (\Cref{lem:technical lemma noisy synd component}), second we use \Cref{lem:technical lemma noisy synd component} in order to construct witnesses for $S \subseteq E_\textrm{ls}$ under the assumption that $|\textrm{MaxConn}_{\alpha_0}(E)| \leq \gamma_0 \sqrt{n}$ (\Cref{lem:technical lemma noisy synd}), third we show using percolation arguments that $|\textrm{MaxConn}_{\alpha_0}(E)| \leq \gamma_0 \sqrt{n}$ holds with probability $1 - \reducedProba$ when $E$ and $D$ are local stochastic (\Cref{lem:perco}) and finally we conclude the proof of \Cref{thm:correction} by showing that if there exist witnesses for all $S \subseteq E_\textrm{ls}$ then $E_\textrm{ls}$ is local stochastic.

\begin{lemma}\label{lem:technical lemma noisy synd component}
  We use the notations of \Cref{subsec:notations} with $\delta < 1/16$ and $\beta \in (0; \beta_1)$.
  \\If the remaining error $E \oplus \hat{E}$ is reduced and $|E \oplus\hat{E}| \leq \gamma_0 \sqrt{n}$ then for all $S \subseteq E \oplus\hat{E}$, there is a $c_0$-witness $W$ for $(S,D)$ with the additional constraint $W \subseteq E \oplus \hat{E}$.
\end{lemma}

\begin{proof}
  We set $E_\textrm{R} = E \oplus \hat{E}$ and define $W$ to be all the connected components of $E_\textrm{R}$ in $\cG$ that contain at least one element of $S$. It is clear that $W \in \mathcal{M}(S)$ and that $W \subseteq E \oplus \hat{E}$ hold, and it remains to prove $|W| \leq c_0 |D \cap \Gamma_X(W)|$.
  \\By locality (\Cref{lem:local} applied with $K = W$), no flip is done by \Cref{algo decodage qec} on the input $E_\textrm{R} \cap W, D \cap \Gamma_X(W)$. Moreover, the remaining error $(E \oplus \hat{E}) \cap W$ is reduced as a subset of the reduced error $E \oplus \hat{E}$ (\Cref{lem:subset reduced}) and satisfies $|(E \oplus \hat{E}) \cap W| \leq \gamma_0 \sqrt{n}$. Hence \Cref{coro:error_qubits_syn} states that:
  \begin{align*}
    |W|
    = |(E \oplus \hat{E}) \cap W|
    \leq c_0 |D \cap \Gamma_X(W) \cap \sigma_X((E \oplus \hat{E}) \cap W)|
    \leq c_0 |D \cap \Gamma_X(W)|.
  \end{align*}
\end{proof}

\begin{lemma}\label{lem:technical lemma noisy synd}
  We use the notations of \Cref{subsec:notations} with $\delta < 1/16$ and $\beta \in (0; \beta_1)$.
  \\If $|\textrm{MaxConn}_{\alpha_0}(E \cup D)| \leq \gamma_0 \sqrt{n}$ then there is a reduced error $E_\textrm{ls}$ equivalent to the remaining error $E \oplus \hat{E}$ such that for all $S \subseteq E_\textrm{ls}$ there is a $c_0$-witness for $(S,D)$.
\end{lemma}

\begin{proof}
  We define $E_\textrm{ls}$ as one of the minimal weight errors whose syndrome is the same than the syndrome of the error $E \oplus \hat{E}$. We will show that $E_\textrm{ls}$ and $E \oplus \hat{E}$ are equivalent.
  \\Let $K$ be a connected component of $U \cup E_\textrm{ls}$. By locality (\Cref{lem:local}), there is a valid execution of \Cref{algo decodage qec} on input $(E \cap K, D \cap \Gamma_X(K))$ whose output is $\hat{E} \cap K$ and whose support is $U \cap K$. 
  Using \Cref{linear number of flips synd error} we know that $(U \cap K) \cup (D \cap \Gamma_X(K))$ is a $2 \alpha_0$-subset of $(E \cap K) \cup (D \cap \Gamma_X(K))$. By \Cref{lem:cup alpha subset} and the fact that $|E_\textrm{ls} \cap K| \leq |(E \oplus \hat{E}) \cap K| \leq |U \cap K|$, we know that $K \cup (D \cap \Gamma_X(K))$ is an $\alpha_0$-subset of $(E \cap K) \cup (D \cap \Gamma_X(K))$. Hence $K \cup (D \cap \Gamma_X(K))$ is an $\alpha_0$-subset of $E \cup D$. From the hypothesis $|\textrm{MaxConn}_{\alpha_0}(E \cup D)| \leq \gamma_0 \sqrt{n}$, we conclude that $|K| \leq \gamma_0 \sqrt{n}$. In particular, $E \oplus \hat{E}$ and $E_\textrm{ls}$ have the same syndrome and the weight of $E \oplus \hat{E} \oplus E_\textrm{ls}$ is smaller than the minimal distance thus $E_\textrm{ls} \cap K$ is equivalent to $(E \oplus \hat{E}) \cap K$. Since this is true for all $K$ then $E_\textrm{ls}$ is equivalent to $E \oplus \hat{E}$.
  \\Let $S \subseteq E_\textrm{ls}$, keeping the same notations we will prove that for each $K$, there is a witness $W_K$ for $S_K := S \cap K$ with the additional constraint $W_K \subseteq E_\textrm{ls} \cap K$:
  \\Let $E' = E_\textrm{ls} \oplus \hat{E}$ and let's prove that if we run \Cref{algo decodage qec} on input $(E' \cap K, D \cap \Gamma_X(K))$ then the output $\hat{E}'$ satisfies the hypothesis of \Cref{lem:technical lemma noisy synd component}. The errors $E$ and $E'$ have same syndrome thus the behavior of \Cref{algo decodage qec} is the same on input $(E, D)$ and on input $(E', D)$. In other words there is a valid execution of \Cref{algo decodage qec} on input $(E', D)$ whose output is $\hat{E}' = \hat{E}$ and whose support is $E' \cup F_0 \cup \ldots \cup F_{f-1}$. By locality (\Cref{lem:local}), there is also a valid execution of \Cref{algo decodage qec} on input $(E' \cap K, D \cap \Gamma_X(K))$ whose output is $\hat{E} \cap K$ and whose support is $(E' \cup F_0 \cup \ldots \cup F_{f-1}) \cap K$. The remaining error of the last execution is $(E' \oplus \hat{E}) \cap K = E_\textrm{ls} \cap K$ which is reduced and is of weight smaller than $\gamma_0 \sqrt{n}$. By \Cref{lem:technical lemma noisy synd component}, there is a $c_0$-witness $W_K$ for $S_K$ with $W_K \subseteq E_\textrm{ls} \cap K$. Finally, $W = \biguplus_K W_K$ is a $c_0$-witness for $S$.

\end{proof}

  The last ingredient before proving \Cref{thm:correction} is provided by \Cref{lem:perco} below: the condition $\textrm{MaxConn}_{\alpha_0}(E \cup D) \leq \gamma_0 \sqrt{n}$ needed in \Cref{lem:technical lemma noisy synd} is verified with high probability for local stochastic errors.

\begin{lemma}[$\alpha$-percolation, \cite{fawzi2018efficient}]\label{lem:perco}
  Let $\cG = (V \cup C_X, \cE)$ be the syndrome adjacency graph and $d :=  \maxDeg$ an upper bound on the degrees in $\cG$ (as defined above \Cref{lem:local}). Then for any $\alpha \in (0,1]$, there exists a threshold $p_{\textrm{th}} = p_{\textrm{th}}(\alpha, d) > 0$ such that for any $t \in \bN^*$ and $p < p_{\textrm{th}}$:
    \begin{align*}
      \mathbb{P}[\textrm{MaxConn}_{\alpha}(E) \geq t] \leq \fullProbaii,
    \end{align*}
    where $C = C(\alpha, p, p_{\textrm{th}})$ is a constant and $E \subseteq \cV$ is a random set chosen accordingly to a local stochastic noise of parameter $p$.
\end{lemma}

\begin{proof}[Proof of \Cref{thm:correction}]

   We define the event $\mathsf{succ}$ to be $\textrm{MaxConn}_{\alpha_0}(E \cup D) \leq \gamma_0 \sqrt{n}$. When we choose the error $E$ and $D$ according to a local stochastic noise, \Cref{lem:perco} ensures that $\mathsf{succ}$ holds with probability at least $1-\reducedProba$. We define $E_\textrm{ls}$ as in \Cref{lem:technical lemma noisy synd}. We know that provided $\textrm{MaxConn}_{\alpha_0}(E \cup D) \leq \gamma_0 \sqrt{n}$, there is a $c_0$-witness $W$ for $(S,D)$ for any $S \subseteq E_\textrm{ls}$ and to conclude we will upper bound the probability that such a witness exists.
  \\First of all, for a fixed $W \subseteq V$, $W$ is a witness for $(S,D)$ implies that there is a set $T \subseteq \Gamma_X(W)$ such that $|W| \leq c_0 |T|$ and $T \subseteq D$. Thus the probability that $W$ is a witness for $(S,D)$ is upper bounded by:
  \begin{align*}
    \mathbb{P}[\text{$W$ is a witness for $(S,D)$}]
    \leq \sum_{\substack{T \subseteq \Gamma_X(W):\\|W| \leq c_0 |T|}} \mathbb{P}[T \subseteq D]
    \leq \sum_{\substack{T \subseteq \Gamma_X(W):\\|W| \leq c_0 |T|}} p_{\mathrm{synd}}^{|T|}.
  \end{align*}
  Since the cardinality of $\Gamma_X(W)$ is upper bounded by $d_B |W|$, we have:
  \begin{align*}
    \mathbb{P}[\text{$W$ is a witness for $(S,D)$}]
    \leq \sum_{|T| \geq |W|/c_0} \binom{d_B |W|}{|T|} p_{\mathrm{synd}}^{|T|}.
  \end{align*}
  For $k, n$ integers, we have the well-known upper bound on the binomial coefficient:
  \begin{align*}
    \binom{n}{k} \leq 2^{n h(k/n)}
  \end{align*}
  where $h(x) = - x \log_2(x) - (1-x) \log_2(1-x)$ is the binary entropy function. Since $n \mapsto n h(k/n)$ is decreasing, with $k = |T|$, $n = d_B |W|$ and under the condition $|T| \geq |W|/c_0$, we have:
  \begin{align*}
    \binom{d_B |W|}{|T|}
    \leq \left(\frac{p_1^{c_0}}{p_{\mathrm{synd}}}\right)^{|T|}
    \text{ where }
    p_1 := 2^{d_B h(1/(d_B c_0))} p_{\mathrm{synd}}^{1/c_0}.
  \end{align*}
  Under the condition that $p_{\mathrm{synd}}$ is sufficiently small to guarantee $p_1 < 1$, we get:
  \begin{align*}
    \mathbb{P}[\text{$W$ is a witness for $(S,D)$}]
    \leq \sum_{|T| \geq |W|/c_0} p_1^{c_0 |T|}
    \leq \frac{p_1^{|W|}}{1 - p_1^{c_0}}.
  \end{align*}
  By conditioning on the event $\textrm{MaxConn}_{\alpha_0}(E \cup D) \leq \gamma_0 \sqrt{n}$, we can now bound the probability of $S \subseteq E_\textrm{ls}$ as follows: 
  \begin{align*}
    \mathbb{P}[S \subseteq E_\textrm{ls} | \textrm{MaxConn}_{\alpha_0}(E \cup D) \leq \gamma_0 \sqrt{n}]
    &= \frac{\mathbb{P}[S \subseteq E_\textrm{ls} , \textrm{MaxConn}_{\alpha_0}(E \cup D) \leq \gamma_0 \sqrt{n}]}{\mathbb{P}[\textrm{MaxConn}_{\alpha_0}(E \cup D) \leq \gamma_0 \sqrt{n}]}\\
    &= 2 \cdot \mathbb{P}[S \subseteq E_\textrm{ls} , \textrm{MaxConn}_{\alpha_0}(E \cup D) \leq \gamma_0 \sqrt{n}],
  \end{align*}
  where we used $\mathbb{P}[\textrm{MaxConn}_{\alpha_0}(E \cup D) \leq \gamma_0 \sqrt{n}] \geq 1 - \reducedProba \geq \frac{1}{2}$ (\Cref{lem:perco}). Now using~\Cref{lem:technical lemma noisy synd} and under the condition that $p_{\mathrm{synd}}$ is sufficiently small to guarantee $p_1 e d < 1$ we get
  \begin{align*}
    \mathbb{P}[S \subseteq E_\textrm{ls} | \textrm{MaxConn}_{\alpha_0}(E \cup D) \leq \gamma_0 \sqrt{n}]
    & \leq 2 \cdot \mathbb{P}[\text{$\exists W: W$ is a witness for $(S,D)$}]\\
    & \leq 2\sum_{\substack{W \in \mathcal{M}(S)\\|W| \geq |S|}} \mathbb{P}[\text{$W$ is a witness for $(S,D)$}]\\
    & \leq 2\sum_{\substack{W \in \mathcal{M}(S)\\|W| \geq |S|}} \frac{p_1^{|W|}}{1 - p_1^{c_0}}\\
    & \leq 2\sum_{|W| \geq |S|} \frac{\left(p_1 e d\right)^{|W|}}{d^{|S|} e \left(1 - p_1^{c_0}\right)}\\
    & \leq 2\frac{\left(p_1 e\right)^{|S|}}{e \left(1 - p_1 e d\right) \left(1 - p_1^{c_0}\right)}.
  \end{align*}
  If $p_{\mathrm{synd}}$ is sufficiently small, then we can make this $\leq \left(K p_{\mathrm{synd}}^{1/c_0}\right)^{|S|}$.

\end{proof}

\section{Parallel version for \Cref{algo decodage qec}}
\label{sec:parallel}

At each step in \Cref{algo decodage qec}, we flip a subset of $F_g \subseteq \Gamma_X(g)$ for some $X$-type generator $g$.  We can parallelize this procedure by flipping at each step the subsets of $F_g$ for a large number of $X$-type generators. However, we have to pay attention to the fact that the $\sigma_X(F_g)$ could intersect. In order to avoid that, we introduce a coloring of the $X$-type generators: if $g_1$ and $g_2$ have the same color then for any $F_1 \subseteq \Gamma_X(g_1)$ and $F_2 \subseteq \Gamma_X(g_2)$: $\sigma_X(F_1) \cap \sigma_X(F_2) = \varnothing$.

\begin{lemma}\label{partition}
  The set of $X$-type generators can be partitioned in less than $\chi$ sets of generators ($\chi$ from \Cref{subsec:notations}) such that each subset can be decoded in parallel.
  \\In other words, there is a partition $C_Z = \biguplus_{k=1}^{\chi} C_Z^k$ of the $X$-type generators such that for any $v_1, v_2 \in V$ and for any $c_1, c_2 \in C_Z$, if $v_1$ and $v_2$ are in the support of the same $Z$-type generator, $v_1$ is in the support of $c_1$ and $v_2$ is in the support of $c_2$ then $c_1$ and $c_2$ cannot be in same subset $C_Z^k$.
\end{lemma}

\begin{proof} 
  We explicitly construct a partition of the $X$-type generators consisting of a constant number of sets, each containing generators with non-intersecting syndromes. From the two factor graphs $G_X$ and $G_Z$, we construct the graph $G_0$ whose vertex set is $C_Z$ the set of $X$-type generators and whose incidence relation is given for $c_1 \in C_Z$ by:
  \begin{align*}
    \Gamma(c_1) = \left\{c_2 \in C_Z: \Gamma_X(\Gamma_Z(c_1)) \cap \Gamma_X(\Gamma_Z(c_2)) \neq \varnothing\right\}.
  \end{align*}
  In words, two $X$-type generators $c_1$ and $c_2$ are adjacent in $G_0$ when we can find a qubit $v_1$ in the support of $c_1$ and a qubit $v_2$ in the support of $c_2$ which are in the support of the same $Z$-type generator. The sets $C_Z^1, \ldots, C_Z^{\chi}$ are defined as a coloring of the graph $G_0$. The chromatic number is upper bounded by the degree of the graph (plus 1) thus it is upper bonded by $\chi$ from \Cref{subsec:notations}.
\end{proof}

This leads to \Cref{algo decodage qec parallel} below which is a parallelized version of \Cref{algo decodage qec}. It is important to note a difference with \Cref{algo decodage qec} though: instead of running until no flips reduce the syndrome weight, \Cref{algo decodage qec parallel} runs for a fixed number of steps and may have some (final) steps that do not reduce the syndrome weight.

\begin{algorithm}[H]
  \caption{: Parallel small-set-flip decoding algorithm for quantum expander codes of parameter $\beta \in (0; 1]$ and with $f \in \bN$ steps.
  }\label{algo decodage qec parallel}
    {\bf INPUT:} $\sigma \subseteq {C}_X$ a syndrome where $\sigma = \sigma_X(E) \oplus D$ with $E \subseteq {V}$ an error on qubits and $D \subseteq C_X$ an error on the syndrome 
    \\{\bf OUTPUT:} $\hat{E}\subseteq {V}$, a guess for the error pattern (alternatively, a set of qubits to correct)

    \hrule
    \begin{algorithmic}
      \State{$\displaystyle \hat{E}_0 = 0$ ; $\displaystyle \sigma_0 = \sigma$}
      \For{$i \in \llbracket 0; f - 1 \rrbracket$}
      \State{$k = i \mod{\chi}$ \text{\qquad // Current color}}
      \ForP{$\displaystyle g \in C_Z^k$}
      \If{$\displaystyle \cF_g := \left\{F \in \cF: F \subseteq \Gamma_Z(g) \And \Delta(\sigma_i, F) \geq \beta |\sigma_X(F)|\right\} \neq \varnothing$}
      \State{$\displaystyle F_g =$ any element of $\cF_g$}
      \Else{\State{$\displaystyle F_g = \varnothing$}}
      \EndIf
      \EndForP
      \State{$\displaystyle F_i = \bigoplus_{g \in C_{Z}^k} F_g$}
      \State{$\displaystyle \hat{E}_{i+1} = \hat{E}_i \oplus F_i$}
      \State{$\displaystyle \sigma_{i+1} = \sigma_i \oplus  \sigma_X (F_i)$ \qquad // $\displaystyle \sigma_{i+1} = \sigma_X (E \oplus \hat{E}_{i+1}) \oplus D$}
      \EndFor

    \end{algorithmic}
\end{algorithm}

\begin{remark}
  \label{rq:eq algo}
  The sequence of flips performed by \Cref{algo decodage qec parallel} could have been performed by \Cref{algo decodage qec}. As a consequence, if we define $U = E \cup F_0 \cup \ldots \cup F_{f-1}$ the execution support then \Cref{linear number of flips synd error} implies that the set $U \cup D$ is $2 \alpha_0$-subset of $E \cup D$.
\end{remark}

If we fix the number of steps of \Cref{algo decodage qec parallel} to be $f = f_0(|\sigma|)$ where $f_0$ is defined in \Cref{subsec:notations}, then it is clear that the complexity is logarithmic in the size of the input. The tricky part is to show that the output is local stochastic and the proof will follow the same scheme than the proof for \Cref{algo decodage qec}: \Cref{lem:np small witness} is the counterpart of \Cref{lem:technical lemma noisy synd component}, \Cref{lem:np witness} is the counterpart of \Cref{lem:technical lemma noisy synd} and \Cref{thm:np correction} is the counterpart of \Cref{thm:correction}.

\begin{theorem}\label{thm:np correction}
  We use the notations of \Cref{subsec:notations} with $\delta < 1/16$, $\beta \in (0; \beta_1)$ and $c > c_2 + 1$ running \Cref{algo decodage qec parallel} with $f = f_0(|\sigma|)$ steps instead of \Cref{algo decodage qec}.
  \\There exist constants $p_0 > 0, p_1 > 0$ such that the following holds. Suppose the pair $(E, D)$ satisfies a local stochastic noise model with parameter $(p_{\mathrm{phys}}, p_{\mathrm{synd}})$ where $p_{\mathrm{phys}} < p_0$ and $p_{\mathrm{synd}} < p_1$. Then there exists an event $\mathsf{succ}$ that has probability $1 - \reducedProba$ and a random variable $E_{\textrm{ls}}$ that is equivalent to $E \oplus \hat{E}$ such that conditioned on $\mathsf{succ}$, $E_{\textrm{ls}}$ has a local stochastic distribution with parameter $p_{\mathrm{ls}} = K p_{\mathrm{synd}}^{1/c_3}$ where $K$ is a constant independent of $p_{\mathrm{synd}}$.
\end{theorem}

About the parameter $c$, \Cref{algo decodage qec parallel} depends on it because of the number of steps $f = f_0(|\sigma|)$. When $c$ increases, $\eta$ decreases and the number of steps $f_0(|\sigma|)$ decreases but the drawback is that $c_3$ increases and thus the local stochastic parameter $p_{\mathrm{ls}}$ associated to the error $E_{\textrm{ls}}$ gets worse. On the other hand, when $c$ is chosen close to $c_2 + 1$, $p_{\mathrm{ls}}$ gets better but the number of steps grows.

The proof of \Cref{thm:np correction} proceeds in the same way as the proof of \Cref{thm:correction}. Indeed, it is sufficient to find a reduced error $E_\textrm{ls}$ equivalent to $E \oplus \hat{E}$ and such that for all $S \subseteq E_\textrm{ls}$ there is a $c_3$-witness for $(S,D)$ in the sense of \Cref{def:witness}. As for the sequential algorithm, we start by establishing the existence of a witness under the assumptions that the remaining error $E \oplus \hat{E}$ is reduced with weight smaller than $\gamma_0 \sqrt{n}$. This will be done in \Cref{lem:np small witness} below, but we will first establish two useful properties of \Cref{algo decodage qec parallel}: it is local (\Cref{lem:np locality}) and at any step $i$ of the algorithm, we either have $|\sigma_i| < c |D|$ or the syndrome weight will decrease rapidly in the following steps (\Cref{lem:pn synd decrease}).

\begin{lemma}[Locality for \Cref{algo decodage qec parallel}]\label{lem:np locality}
  We use the notations of \Cref{subsec:notations} running \Cref{algo decodage qec parallel} with $f$ steps instead of \Cref{algo decodage qec}.
  \\For any $K \subseteq U$ with $\Gamma(K) \cap \Gamma(U \setminus K) = \varnothing$ in $\cG$ (for the same $\cG$ as \Cref{lem:local}), there is a valid execution of \Cref{algo decodage qec parallel} with $f$ steps on input $(E \cap K, D \cap \Gamma(K))$ which flips $F_0 \cap K, \ldots, F_{f-1} \cap K$.
\end{lemma}

\begin{lemma}\label{lem:pn synd decrease}
  We use the notations of \Cref{subsec:notations} with $\delta < 1/16$, $\beta \in (0; \beta_1)$ and $c > c_2 + 1$ running \Cref{algo decodage qec parallel} with $f$ steps instead of \Cref{algo decodage qec}.
  \\Suppose $|U| \leq \gamma_0 \sqrt{n}$ and that for some $i \in \llbracket 0; f-\chi \rrbracket$ we have $|\sigma_i| \geq c |D|$ then $|\sigma_{i + \chi}| \leq \eta |\sigma_i|$ where $\eta$ is defined in \Cref{subsec:notations}.
\end{lemma}

\begin{proof}
  We define the error $E_i := E \oplus F_0 \oplus \ldots \oplus F_{i-1}$. We have:
  \begin{align}\label{eq1}
    |\sigma_X(E_i)|
    \geq |\sigma_i| - |D|
    \geq (c-1) |D|
    > c_2 |D \cap \sigma_X(E_i)|.
  \end{align}
Since  $|E_i| \leq |U| \leq \gamma_0 \sqrt{n}$, we can apply \Cref{prop:error_qubits_syn} to $E_i$: there exists a color $k \in \llbracket 0;\chi-1 \rrbracket$, there exists $G \subseteq \cF_k := \cF \cap \{\Gamma_X(g): g \in C_Z^k\}$ such that $\Delta(\sigma_i, F) \geq \beta |\sigma_X(F)|$ for all $F \in G$ and:
  \begin{align}\label{eq2}
    & \sum_{F \in G} |\sigma_X(F)|
    \geq \frac{c_1}{\chi}\left[\left|\sigma_X(E_i)\right| - c_2 |D \cap \sigma_X(E_i)|\right].
  \end{align}
  We set $j$ as the integer in $\llbracket i + 1; i + \chi \rrbracket$ such that $j-1$ is a step of the algorithm where we flip the color $k$. In other words, $k = (j-1) \mod \chi$.
  \\We define $G_1 \subseteq \cF$ the set of all flips done during the steps $i, \ldots, j-2$ and $G_2 \subseteq \cF_{k}$ the set of all flips done at step $j-1$. By the coloring property, $\sigma_X(F) \cap \sigma_X(F') = \varnothing$ for all $F, F' \in G$ thus for all $F_1 \in G_1$, each check-node of $\sigma_X(F_1)$ is in the syndrome of at most one element of $G$. In other words, for a given syndrome $\sigma$ and for $F_1 \in G_1$, the number of $F \in G$ such that $\Delta(\sigma, F) \neq \Delta(\sigma \oplus \sigma_X(F_1), F)$ is at most $|\sigma_X(F_1)|$. As a consequence:
  \begin{align}\label{eq3}
    |G_2| \geq |G| - \sum_{F_1 \in G_1} |\sigma_X(F_1)|.
  \end{align}
  Moreover:
  \begin{align}\label{eq4}
    |\sigma_{i}| - |\sigma_{j-1}|
    \geq \sum_{F_1 \in G_1} \beta |\sigma_X(F_1)|,
  \end{align}
 and
  \begin{align}\label{eq5}
    |\sigma_{j-1}| - |\sigma_{j}|
    \geq \sum_{F_2 \in G_2} \beta |\sigma_X(F_2)|
    \geq \beta |G_2|.
  \end{align}
  Combining \cref{eq3,eq4,eq5} we get:
  \begin{align*}
    \sum_{F \in G} |\sigma_X(F)|
    \leq \sum_{F \in G} d_A d_B
    = d_A d_B |G|
    \leq d_A d_B \left(|G_2| + \sum_{F_1 \in G_1} |\sigma_X(F_1)|\right)
    \leq \frac{d_A d_B}{\beta} (|\sigma_i| - |\sigma_j|).
  \end{align*}
  By \cref{eq2}:
  \begin{align*}
    |\sigma_{i}| - |\sigma_j|
    \geq \frac{\beta}{d_A d_B} \sum_{F \in G} |\sigma_X(F)|
    \geq \frac{\beta c_1}{d_A d_B \chi} \left[\left|\sigma_X(E_i)\right| - c_2 |D \cap \sigma_X(E_i)|\right].
  \end{align*}
  Using $|\sigma_X(E_i)| \geq (c-1) |D| \geq (c-1) |D \cap \sigma_X(E_i)|$ from \cref{eq1} and the hypothesis $|\sigma_i| \geq c |D|$:
  \begin{align*}
    |\sigma_{i}| - |\sigma_j|
    \geq \frac{c (1 - \eta)}{c - 1} \left|\sigma_X(E_i)\right|
    \geq \frac{c (1 - \eta)}{c - 1} (|\sigma_i| - |D|)
    \geq (1 - \eta) |\sigma_i|.
  \end{align*}
  Thus $|\sigma_{i + \chi}| \leq |\sigma_j| \leq \eta |\sigma_i|$.
\end{proof}

\begin{lemma}\label{lem:np small witness}
  We use the notations of \Cref{subsec:notations} with $\delta < 1/16$, $\beta \in (0; \beta_1)$ and $c > c_2 + 1$ running \Cref{algo decodage qec parallel} with $f = f_0(|\sigma|)$ steps instead of \Cref{algo decodage qec}.
  \\Suppose $|U \cup D| \leq \gamma_0 \sqrt{n}$ and the remaining error $E \oplus \hat{E}$ is reduced then for all $S \subseteq E \oplus \hat{E}$, there is a $c_3$-witness $W$ for $(S,D)$ with $W \subseteq U$.
\end{lemma}

\begin{proof}
  For any $S \subseteq E_\textrm{R} := E \oplus \hat{E}$ and all $i \in \llbracket 0;f\rrbracket$ we define:
  \begin{itemize}
  \item $U_i := (E \oplus F_0 \oplus \ldots \oplus F_{i-1}) \cup F_i \cup \ldots \cup F_{f-1} = E_{\textrm{R}} \cup F_i \cup \ldots \cup F_{f-1}$, 
  \item $W_i$, the union of the connected components of $U_i$ which contains at least one element of $S$,
  \item $E_i := (E \oplus F_0 \oplus \ldots \oplus F_{i-1}) \cap W_i$,
  \item $X_i := \Gamma(W_i)$,
  \item $s_i := \sigma_i \cap X_i = \sigma_X(E_i) \oplus (D \cap X_i)$.
  \end{itemize}
  Note that $U_{i+1} \subseteq U_i, W_{i+1} \subseteq W_i$ and $X_{i+1} \subseteq X_i$.
  \\We are going to show that there is at least one $i \in \llbracket 0;f\rrbracket$ such that $W_i$ is a witness for $S$.
  Suppose that for some $i \in \llbracket 0;f-\chi\rrbracket$, we have $|s_i| \geq c |D \cap X_i|$. By locality (\Cref{lem:np locality} applied with $f - i$ steps, for $K = W_i$ and errors $E \oplus F_0 \oplus \ldots \oplus F_{i-1}, D$), there is a valid execution of \Cref{algo decodage qec parallel} with $f - i$ steps on input $(E_i, D \cap X_i)$ which flips $F_i \cap W_i, \ldots, F_{f-1} \cap W_i$. Note that the remaining error is $E_i \oplus (F_i \cap W_i) \oplus \ldots \oplus (F_{f-1} \cap W_i) = E_\textrm{R} \cap W_i$. Hence for $k \in \llbracket i;f-1 \rrbracket$:
  \begin{align}\label{eq6}
    \Delta(\sigma_k \cap X_i, F_k \cap W_i) \geq \frac{\beta d_A}{2} |F_k \cap W_i|.
  \end{align}
  And by \Cref{lem:pn synd decrease}:
  \begin{align*}
    |\sigma_{i + \chi} \cap X_i| \leq \eta |\sigma_{i} \cap X_i|.  
  \end{align*}
  Hence using $X_{i+\chi} \subseteq X_i$:
  \begin{align*}
    |s_{i+\chi}|
    = |\sigma_{i + \chi} \cap X_{i+\chi}|
    \leq |\sigma_{i + \chi} \cap X_{i}|
    \leq \eta |\sigma_{i} \cap X_i|
    = \eta |s_i|.
  \end{align*}
  The last inequality is true as soon as $|s_i| \geq c |D \cap X_i|$ but we run \Cref{algo decodage qec parallel} with $f = f_0(|\sigma|) \geq f_0(|s_0|)$ steps hence there is at least one $i \in \llbracket 0;f\rrbracket$ such that:
  \begin{align}\label{eq7}
    |s_i| \leq c |D \cap X_i|.
  \end{align}

  $E_\textrm{R} \cap W_i$ is reduced as a subset of the reduced error $E_\textrm{R}$ (\Cref{lem:subset reduced}) thus we apply \Cref{lem:robust} to get:
  \begin{align}\label{eq8}
    |\sigma_X(E_\textrm{R} \cap W_i)| \geq \frac{\beta_0 d_A}{2} |E_\textrm{R} \cap W_i|  
  \end{align}
  As a conclusion, \cref{eq6,eq7,eq8} imply that $W_i$ is a $c_3$-witness for $S$:
  \begin{align*}
    |W_i|
    & = |W_i \cap U_i|\\
    & \leq |E_\textrm{R} \cap W_i| + \sum_{k = i}^{f-1} |F_k \cap W_i|\\
    & \leq \frac{2}{\beta_0 d_A} |\sigma_X(E_\textrm{R} \cap W_i)| + \frac{2}{\beta d_A}\sum_{k = i}^{f-1} \Delta(\sigma_k \cap X_i, F_k \cap W_i) \text{\qquad by \cref{eq6,eq8}}\\
    & \leq \frac{2}{\beta d_A} \left(|\sigma_X(E_\textrm{R} \cap W_i)| + |\sigma_i \cap X_i| - |\sigma_{f} \cap X_i|\right)\\
    & = \frac{2}{\beta d_A} \left(|\sigma_X(E_\textrm{R} \cap W_i)| + |s_i| - |\sigma_X(E_\textrm{R} \cap W_i) \oplus (D \cap X_i)|\right)\\
    & \leq \frac{2}{\beta d_A} \left(|D \cap X_i| + |s_i|\right)\\
    & \leq c_3 |D \cap X_i| \text{\qquad by \cref{eq7}}
  \end{align*}
\end{proof}

\begin{lemma}\label{lem:np witness}
  We use the notations of \Cref{subsec:notations} with $\delta < 1/16$, $\beta \in (0; \beta_1)$ and $c > c_2 + 1$ running \Cref{algo decodage qec parallel} with $f = f_0(|\sigma|)$ steps instead of \Cref{algo decodage qec}.
  \\If $|\textrm{MaxConn}_{\alpha_0}(E \cup D)| \leq \gamma_0 \sqrt{n}$ then there is a reduced error $E_\textrm{ls}$ equivalent to the remaining error $E \oplus \hat{E}$ such that for all $S \subseteq E_\textrm{ls}$ there is a $c_3$-witness for $(S,D)$.
\end{lemma}
\begin{proof}
Using the locality property (\Cref{lem:np locality}), the proof is identical to the proof of \Cref{lem:technical lemma noisy synd}.
\end{proof}

With \Cref{lem:np witness}, the proof of \Cref{thm:np correction} is the same as the proof \Cref{thm:correction}.

%%%%%%%%%%%%%%%%%%%%%%
\subsection*{Acknowledgments}

  We would like to thank Benjamin Audoux, Alain Couvreur, Anirudh Krishna, Vivien Londe, Jean-Pierre Tillich and Gilles Z\'emor for many fruitful discussions on quantum codes as well as Daniel Gottesman for answering our questions about \cite{gottesman2014fault}. We would also like to thank the anonymous reviewers for their useful comments.
  AG and AL acknowledge support from the ANR through the QuantERA project QCDA.

%\bibliographystyle{plain}
%\bibliography{QEC}

\end{document}